\documentclass[journal,onecolumn, 12pt]{IEEEtran}
\IEEEoverridecommandlockouts
\usepackage{cite}
\usepackage{amsmath,amssymb,amsfonts, mathtools}
\usepackage{algorithmic}
\usepackage{graphicx}
\usepackage[left=1in, right=1in, top=1in, bottom=1in]{geometry}
\usepackage{textcomp}
\usepackage{xcolor}
\usepackage{algorithmic}
\usepackage{algorithm}
\usepackage{float}
\usepackage{booktabs}
\usepackage{caption}
\usepackage{siunitx}
\usepackage{tabularx}
\usepackage{adjustbox}
\usepackage{braket}
\usepackage{color}
\usepackage{colortbl}
\usepackage{placeins}
\usepackage{subfigure}
\usepackage{amsthm,amssymb}
\numberwithin{figure}{section}

\newtheorem{theorem}{Theorem}
\newtheorem{lemma}{Lemma}
\theoremstyle{definition}
\newtheorem{example}{Example}

\newtheorem{property}{Property}

\newtheorem{remark}{Remark}
\newtheorem{definition}{Definition}

\usepackage{tikz,xcolor,hyperref}

\definecolor{lime}{HTML}{A6CE39}
\DeclareRobustCommand{\orcidicon}{%
	\begin{tikzpicture}
		\draw[lime, fill=lime] (0,0) 
		circle [radius=0.16] 
		node[white] {{\fontfamily{qag}\selectfont \tiny ID}};
		\draw[white, fill=white] (-0.0625,0.095) 
		circle [radius=0.007];
	\end{tikzpicture}
	\hspace{-2mm}
}

\foreach \x in {A, ..., Z}{%
	\expandafter\xdef\csname orcid\x\endcsname{\noexpand\href{https://orcid.org/\csname orcidauthor\x\endcsname}{\noexpand\orcidicon}}
}

\def\BibTeX{{\rm B\kern-.05em{\sc i\kern-.025em b}\kern-.08em
		T\kern-.1667em\lower.7ex\hbox{E}\kern-.125emX}}
\begin{document}
	\title{Decoding Algorithms for Two-dimensional Constacyclic Codes over $\mathbb{F}_q$}
	
	\author{Vidya~Sagar$^{\ast}$\orcidA{},~Shikha~ Patel\orcidB{}~and~ Shayan Srinivasa Garani\orcidC{}

		\thanks{Vidya Sagar, Shikha Patel and Shayan Srinivasa Garani are with the Department of Electronic Systems Engineering, Indian Institute of Science, Bengaluru, Karnataka 560012, India (e-mail: vsagariitd@gmail.com; shikha\_1821ma05@iitp.ac.in; shayangs@iisc.ac.in).$^{\ast}$Corresponding author}
	}
	
	\maketitle
	
	\begin{abstract}
		We derive the spectral domain properties of two-dimensional (2-D) $(\lambda_1, \lambda_2)$-constacyclic codes over $\mathbb{F}_q$ using the 2-D finite field Fourier transform (FFFT). Based on the spectral nulls of 2-D $(\lambda_1, \lambda_2)$-constacyclic codes, we characterize the structure of 2-D constacyclic coded arrays. The proposed 2-D construction has flexible code rates and works for any code areas, be it odd or even area. We present an algorithm to detect the location of 2-D errors. Further, we also propose decoding algorithms for extracting the error values using both time and frequency domain properties by exploiting the sparsity that arises due to duality in the time and frequency domains. Through several illustrative examples, we demonstrate the working of the proposed decoding algorithms. 
		
	\end{abstract}
	\begin{IEEEkeywords}
		Finite field, two-dimensional constacyclic codes, two-dimensional finite field Fourier transform, decoding. 
	\end{IEEEkeywords}
	\section{Introduction}
	Cyclic codes are one of the most important classes of linear codes due to their rich algebraic structure. These codes represent a significant category of codes for two primary reasons. First, cyclic codes can be effectively encoded with shift registers, which is why they are favored for their application in engineering. Second, from a theoretical point of view, cyclic codes are readily identified as the ideals of the quotient ring $\mathbb{F}_q[x]/\langle x^n-1\rangle$.
	It was first studied in 1957 by Prange \cite{prange1957cyclic}. Later, it was generalized as constacyclic codes by Berlekamp \cite{berlekamp1968algebraic}. These codes can correct both random errors and burst errors. Cyclic codes based on Bose-Chaudhuri-Hocquenghem (BCH) and Reed-Solomon (RS) designs have made a tremendous impact in data storage systems, such as hard disk drives and flash memories. Hence, it is natural to consider two-dimensional (2-D) generalizations of 1-D cyclic codes for investigation with applications to multidimensional data storage devices, bar code readers, etc.\par 
	
	In a seminal work done fifty years ago, Imai \cite{imaiTDC} introduced a theory for binary 2-D cyclic (TDC) codes of odd area. These binary TDC codes were derived using the idea of a common zero (CZ) set, i.e., zeros common to all codewords. Further, parity-check tensors, parity-check positions, and dual of TDC codes were characterized using the CZ set. TDC codes can be practically implemented using simple 2-D linear feedback shift registers and combinatorial logics. As a result, the examination of these codes over finite fields has received considerable attention from other researchers in the recent past. In \cite{li20142}, the authors explored two-dimensional skew cyclic codes over a finite field. Two-dimensional cyclic codes corresponding to the ideals of the ring $\mathbb{F}[x,y]/\langle x^l-1,y^{2^k}-1\rangle$ were explored by the authors in \cite{sepasdar2016}. Later, repeated root two-dimensional constacyclic codes of different lengths were explored by the authors in \cite{rajabi2018repeated,patel2020repeated}. The structural behavior of 2-D skew cyclic codes over the ring $\mathbb{F}_q+u\mathbb{F}_q$ with $u^2=1$ was explored by Sharma et al. \cite{sharma2019class}. Recently, Patel and Prakash \cite{Patel_2Dcode} have extended this work to $(\lambda_1,\lambda_2)$-constacyclic codes. 
	
	In almost all of the aforementioned works, the authors have used the polynomial decomposition method and not the common zero sets approach as in Imai's work \cite{imaiTDC}. Characterization of CZ sets is useful in the encoding and decoding of these codes. Taking into account these practical code design considerations such as code rates, flexible code areas, etc., we generalized Imai's theory on TDC codes to 2-D $(\lambda_1,\lambda_2)$-constacyclic codes using the CZ set approach with the motivation that these codes can get a better minimum distance compared to 2-D cyclic codes with the same code areas and code rates, thus improving the error correction ability in \cite{CCDS}. An encoding scheme for 2-D
	constacyclic codes over $\mathbb{F}_q$ was also provided in that work \cite{CCDS}.\par
	
	Frequency domain (image under the 2-D finite field Fourier transform) coding is also called transform domain coding or spectral domain coding. Transform domain coding techniques for finite fields were originally introduced by Blahut \cite{blahut1979algebraic, blahut1983theory}, who first examined cyclic BCH codes and their decoding using syndrome computations \cite{blahut1979transform}. Transform domain encoders and decoders are simple to implement. Further, the syndromes can be readily identified by non-zero spectral values in the parity positions. Inspired by the transform domain approach, several researchers have explored coding techniques in this area; the readers can refer to the following papers \cite{rajan1992transform, rajan1994generalized, rajan1994transform, massey1998discrete, dey2003dft} and references therein.\par 
	
	Decoding of 2-D BCH codes was originally introduced by Madhusudana and Siddiqui by extending Blahut's 1-D approach to the 2-D case using row and column decoders. This approach was made more efficient by exploiting the conjugacy properties of finite fields along with parallel architectures for the optimization of the encoding and decoding algorithms towards low-latency and high-throughput hardware architectures in \cite{mondal2021efficient}. Further, transform domain coding for handling specific 2-D error patterns was explored by Yoon and Moon by extending Imai's approach using the disjoint syndrome criterion for handling error patterns. The reader must note that, in general, it is difficult to endow a 2-D code with all the properties, such as correction of random errors, cluster errors, and specific error patterns. For example, Abdel-Ghaffar, McEliece and Van Tilborg \cite{khaled} have specifically investigated the construction of 2-D codes for rectangular cluster error identification. The interested author can refer to \cite{khaled} and references therein.\par  
	
	Our approach in this article is not in the line of specific codes for burst identification, but rather focused on 2-D random errors, since specific burst patterns can be treated as random errors by designing codes over higher dimensions in appropriate field extensions, following the constructions in \cite{CCDS}. In this article, we focus on characterizing the 2-D code properties of constacyclic codes and constructing decoding algorithms to handle random errors. The central ideas behind our approach are as follows:\\ 
	(a) Based on the common zero (CZ) set, we characterize the code properties of 2-D constacyclic codes. Further, we derive several useful spectral domain properties of 2-D contacyclic coded arrays, useful during decoding.\\
	(b) We identify the error locations by computing the syndromes over the row and column locations, and identify the corresponding bi-index of the error locations.\\ 
	(c) We propose decoding algorithms for both low rate and high rate 2-D constacyclic codes by exploiting the sparsity due to duality in the time and frequency domains based on the CZ sets and spectral nulls, respectively. \par 
	
	Here are the major contributions of our work:
	\begin{enumerate}
		\item We discuss properties of 2-D FFFT, 2-D constacyclic codes in frequency domain, and spectral nulls of 2-D constacyclic codes over $\mathbb{F}_q$.
		\item We present an algorithm for error-detection in 2-D constacyclic codes over $\mathbb{F}_q$. Further, we propose a method for finding the locations of errors.
		\item We propose two efficient decoding methods for correcting random errors in 2-D constacyclic codes over $\mathbb{F}_q$, exploiting the sparsity resulting from the duality of finite field Fourier transform.
	\end{enumerate}
	
	The rest of the paper is organized as follows. In Section \ref{section2}, we briefly review the mathematical background on constacyclic codes  required for the subsequent sections. In Section \ref{section3}, we discuss finite field Fourier transforms and their properties for 2-D constacyclic codes. In Section \ref{section4}, we fully characterize the properties of 2-D constacyclic coded arrays. We introduce the notion of common zero sets for 2-D constacyclic codes over $\mathbb{F}_q$ leading to generator and parity check tensors along with illustrative examples. In Section \ref{section5}, we discuss a technique for error detection in 2-D constacyclic codes over $\mathbb{F}_q$ followed by a method to find the locations of errors. In addition, we propose efficient decoding methods for correcting random errors in 2-D constacyclic codes over $\mathbb{F}_q$. Algorithms for these methods are also presented in this section. Several illustrative examples are worked out to show how the proposed algorithm corrects random within the design distance. Section \ref{section6} concludes the paper.

	\section{Preliminaries} \label{section2}
	In this section, we discuss some basic definitions and results that are useful for the rest of the paper. Let $\mathbb{F}_{q}$ denote the finite field of order $q$, where $q$ is some power of a prime $p$, and let $\mathbb{F}_q^{\ast}=\mathbb{F}_q\setminus \{0\}$.
	\subsection{Notations and background}
	\begin{definition}
		An $[n, k]$-linear code $\mathcal{C}$ over $\mathbb{F}_q$ is called cyclic code if $(c_{n-1}, c_0, \dots , c_{n-2}) \in \mathcal{C}$ whenever \\
		$(c_0, c_1, \dots , c_{n-1}) \in \mathcal{C}$.
	\end{definition}
	By identifying any vector $(c_0, c_1, \dots , c_{n-1}) \in \mathbb{F}_q^n$ with the corresponding polynomial\\
	$$c_0 + c_1 x+ c_2 x^2 + \cdots + c_{n-1} x^{n-1} \in \mathbb{F}_q[x]/\langle x^n - 1\rangle,$$ one can observe that a cyclic code $\mathcal{C}$ of length $n$ over $\mathbb{F}_q$ is nothing but an ideal of the quotient ring $\mathbb{F}_q[x]/\langle x^n - 1\rangle$. 
	\begin{definition}
		Let $\lambda \in \mathbb{F}_q^{\ast}$. An $[n, k]$-linear code $\mathcal{C}$ over $\mathbb{F}_q$ is called $\lambda$-constacyclic if $(\lambda c_{n-1}, c_0, \dots , c_{n-2}) \in \mathcal{C}$ whenever $(c_0, c_1, \dots , c_{n-1}) \in \mathcal{C}$.     
	\end{definition}
	Observe that $\lambda$-constacyclic codes is one of the generalizations of cyclic codes; in other words, cyclic codes of length $n$ over $\mathbb{F}_q$ are a subclass of $\lambda$-constacyclic codes of length $n$ over $\mathbb{F}_q$. In particular, when $\lambda = 1$, the $\lambda$-constacyclic code coincides with the cyclic code. \par
	Recall that for $c= (c_0,c_1,...,c_{n-1})\in\mathcal{C}\subseteq  \mathbb{F}_q^n$, the Hamming weight $w_H(c)$ is equal to the number of non-zero components of $c$, and for any two codewords $c$ and $c'$ of $\mathcal{C},$ the Hamming distance between $c$ and $c'$ is defined as $d_H(c,c')=w_H(c-c')$. The Hamming distance for the code $\mathcal{C}$ is $d= \mathrm{\min}\{ d(c,c')~|~ c, ~c'\in \mathcal{C},  c \neq c' \}.$ Let $c=(c_0,c_1,...,c_{n-1})$ and $c'=(c'_0,c'_1,\dots,c'_{n-1})$ be two elements of $\mathcal{C}$. Then, the (Euclidean) inner product of $c$ and $c'$ in $\mathbb{F}_q^n$ is defined as 
	$$c\cdot c' = \sum_{j=0}^{n-1}c_jc'_j \in \mathbb{F}_q.$$
	The dual code of $\mathcal{C}$ is $\mathcal{C}^\perp= \{c\in \mathbb{F}_q^n~|~ c\cdot c' = 0, \forall ~c'\in \mathcal{C} \}$. The code $\mathcal{C}$ is said to be self-orthogonal if $\mathcal{C}\subseteq \mathcal{C}^\perp$ and self-dual if $\mathcal{C}=\mathcal{C}^\perp$.\par
	Let the set of all $M\times N$ arrays over $\mathbb{F}_q$ be denoted by $\mathbb{M}_{M\times N}(\mathbb{F}_q)$. We will represent an array of size $M \times N$ over $\mathbb{F}_q$ using its bi-variate polynomial. Let 
	\begin{equation}\label{Omega}
		\Omega = \{(i, j) \vert \ 0\leq i \leq M-1, 0\leq j \leq N-1\}
	\end{equation} 
	and let 
	\begin{equation}
		\mathcal{P}[\Omega]=\Big\{\sum_{(i,j)\in \Omega}p_{i,j}x^{i}y^{j} \ \vert \ p_{i,j} \in \mathbb{F}_q\Big\}.
	\end{equation} 
	For an array $c=(c_{i,j}) \in \mathbb{M}_{M\times N}(\mathbb{F}_q) $, the corresponding bi-variate polynomial is
	\begin{equation}
		c(x,y)= \sum_{(i,j)\in \Omega}c_{i,j}x^{i}y^{j}.
	\end{equation}
The array $c$ will be represented by its polynomial form $c(x,y)$, and conversely, whenever necessary.
	\begin{definition}
		A two-dimensional code $\mathcal{C}$ over $\mathbb{F}_q$ of area $M\times N$ is a non-empty subset of $\mathbb{M}_{M\times N}(\mathbb{F}_q)$.
	\end{definition}
	\begin{definition}
		If $\mathcal{C}$ is a subspace of the $MN$-dimensional vector space $\mathbb{M}_{M\times N}(\mathbb{F}_q)$ over $\mathbb{F}_q$, then $\mathcal{C}$ is called two-dimensional (2-D) linear code of area $MN$ over $\mathbb{F}_q$.
	\end{definition}
	
	\begin{definition}
		Let $\lambda_1, \lambda_2 \in \mathbb{F}_q^{\ast}$. A 2-D linear code is called $(\lambda_1, \lambda_2)$-constacyclic if $\{xc(x, y)\}_{\Omega}, \{yc(x, y)\}_{\Omega} \in  \mathcal{C}$ for each $c(x, y) \in \mathcal{C}$, where 
		\begin{equation*}
			\{f(x, y)\}_{\Omega}\cong f(x, y) \textnormal{ mod} (x^M-\lambda_1, y^N-\lambda_2).
		\end{equation*}    
	\end{definition}
	These representations provide an explicit algebraic description for two-dimensional linear codes.
	\begin{itemize}

		\item Let $\lambda_1\in \mathbb{F}_q^*$. Then, $\mathcal{C}$ is said to be a column $\lambda_1$-constacyclic code of area $MN$ if for every $M\times N$ array $c=(c_{ij})\in \mathcal{C},$ we have its column-shift
		\[ c_{\lambda_1}=
		\begin{pmatrix}
			\lambda_1c_{M-1,0} & \lambda_1c_{M-1,1} & \dots & \lambda_1c_{M-1,N-1}\\
			c_{0,0} & c_{0,1}  & \dots & c_{0,N-1}\\
			\vdots & \vdots & \ddots & \vdots\\
			c_{M-2,0} & c_{M-2,1} & \dots & c_{M-2,N-1}
		\end{pmatrix}\in \mathcal{C}.
		\]
		\item Let $\lambda_2\in \mathbb{F}_q^*$. Then, $\mathcal{C}$ is said to be a row $\lambda_2$-constacyclic code of area $MN$ if for every $M\times N$ array $c=(c_{ij})\in \mathcal{C},$ we have its row-shift
		\[ c_{\lambda_2}=
		\begin{pmatrix}
			\lambda_2c_{0,N-1} & c_{0,0}  & \dots & c_{0,N-2}\\
			\lambda_2c_{1,N-1} & c_{1,0}  & \dots & c_{1,N-2}\\
			\vdots & \vdots & \ddots & \vdots\\
			\lambda_2c_{M-1,N-1} & c_{M-1,0} & \dots & c_{M-1,N-2}
		\end{pmatrix}\in \mathcal{C}.\]
	\end{itemize}
	\begin{definition}
				If $\mathcal{C}$ is both column $\lambda_1$-constacyclic and row $\lambda_2$-constacyclic, then $\mathcal{C}$ is said to be a 2-D $(\lambda_1,\lambda_2)$-constacyclic code of area $MN$.\\ Clearly, when $\lambda_1= \lambda_2=1,$ 2-D $(\lambda_1,\lambda_2)$-constacyclic code coincides with the two-dimensional cyclic (TDC) code~\cite{imaiTDC}.
	\end{definition}
	It is noted that there is a one-to-one correspondence between cyclic codes of length $M$ over $\mathbb{F}_q$ and ideals of the quotient ring $\mathbb{F}_q[x]/\langle x^M-1\rangle$. Similarly, in the case of TDC codes, there is also a one-to-one correspondence between TDC codes of area $MN$ over $\mathbb{F}_q$ and ideals of the quotient ring $\mathbb{F}_q[x,y]/\langle x^M-1,y^N-1\rangle$. Hence, a TDC code $\mathcal{C}\subseteq \mathbb{F}_q^{M\times N}$ of area $MN$ over $\mathbb{F}_q$ can be viewed as an ideal of the quotient ring $\mathbb{F}_q[x,y]/\langle x^M-1,y^N-1\rangle$.\par 
	Further, for $\lambda_1, \lambda_2 \in \mathbb{F}_q^*$, a 2-D $(\lambda_1,\lambda_2)$-constacyclic code $\mathcal{C}\subseteq \mathbb{F}_q^{M\times N}$ of area $MN$ over $\mathbb{F}_q$ is an ideal of the quotient ring $\mathbb{F}_q[x,y]/\langle x^M-\lambda_1,y^N-\lambda_2\rangle$. Throughout the paper, we will assume that both $M$ and $N$ are positive integers relatively prime to $p$, the characteristic of $\mathbb{F}_q$ [Ref. p. 3 \cite{imaiTDC}].\par
	\begin{definition}\cite{imaiTDC}
		Let $\mathcal{C}$ be a 2-D $(\lambda_1, \lambda_2)$-constacyclic code of area $MN$. If $\mathcal{B} = \{g_1(x, y), g_2(x, y), \dots, g_t(x, y)\}\subseteq \mathcal{C}$ generates the code $\mathcal{C}$ then $\mathcal{B}$ is called an ideal basis of $\mathcal{C}$.
	\end{definition}
	Observe that ideal basis is not unique. Now, we recall a few definitions and notions. For more details, the readers are referred to \cite{CCDS}.
	Define
	\begin{equation}
		V_{\circ}:=\{(a, b) \vert \ a^M = \lambda_1, b^N=\lambda_2 \}.
	\end{equation}
	Let $\gamma, \beta$ be primitive $M^{\mathrm{th}}$ and $N^{\mathrm{th}}$ roots of $\lambda_1$ and $\lambda_2$, respectively. Note that $\gamma$ and $\beta$ are in an extension field of $\mathbb{F}_q$. Then 
	\begin{equation}
		V_{\circ}=\{(\gamma^{q^i}, \beta^{q^j}) \vert \ 0 \leq i\leq  z_1-1,\ 0\leq j \leq z_2-1 \text{ such that } \gamma^{q^{z_1}}=\gamma, \beta^{q^{z_2}}=\beta \}.
	\end{equation}
	Let $V$ be a set of common roots of some arbitrary polynomials in $\mathbb{F}_q[x, y]$. If $(a, b) \in V$ then $(a^{q^i}, b^{q^i}) \in V$ for $1\leq i \leq z-1$ where $z$ is the least positive integer such that $(a, b) = (a^{q^z}, b^{q^z})$. These points are the only points in the conjugate point set of $(a, b)$. Two distinct elements $\xi_1$ and $\xi_2$ are said to be conjugate with respect to $\mathbb{F}_{q^m}$ if there exists a positive integer $k$ such that $\xi_1 = \xi_2^{q^{mk}}$. Note that $\xi_1$ and $\xi_2$ are roots of the same minimal polynomial over $\mathbb{F}_{q^m}$. By choosing one point from each conjugate set in $V_{\circ}$, we can construct a subset of $V_{\circ}$ such that the first components of any two points in this subset are not conjugates with respect to $\mathbb{F}_{q}$. We denote this subset by $\hat{V}_{\circ}$. The following example illustrates this.
	\begin{example}\label{Eg1}
		Consider the base field $\mathbb{F}_3$. Let $M=4, N=5$ and let $\lambda_1=\lambda_2=2 \in \mathbb{F}_{3}$. Suppose that $\gamma$ and $\beta$ are $4^{\mathrm{th}}$ and $5^{\mathrm{th}}$ primitive roots of $\lambda_1$ and $\lambda_2$, respectively.
		Then the set $V_{\circ}$ is given by the common roots of the polynomials $x^{4}-\lambda_1=0$ and $y^{5}-\lambda_2=0$ as
		\begin{equation*}
			V_{\circ}=\{(\gamma, \beta), (\gamma^3, \beta^3), (\gamma, 2\beta^4), (\gamma^3, 2\beta^2); (\gamma, \beta^3), (\gamma^3, 2\beta^4), (\gamma, 2\beta^2), (\gamma^3, \beta)\}
		\end{equation*}
		and $\widehat{V}_{\circ} = \{(\gamma, \beta), (\gamma, \beta^3)\}.$
	\end{example}
	
	\begin{definition}
		Let $\mathcal{C}$ be a 2-D $(\lambda_1, \lambda_2)$- constacyclic code. Suppose $V_{1}$ is the set of common roots of all the codewords of $\mathcal{C}$. Define $V_{c}:=V_{\circ}\cap V_{1}$ and $\hat{V}_{c}:=\hat{V}_{\circ}\cap V_{1}$. Then $V_{c}$ is called the common zero (CZ) set and $\hat{V}_{c}$ is called the essential common zero (ECZ) set.
	\end{definition}
	\begin{remark}
		Note that the first components of any two points in $\widehat{V}_c$ are not conjugates with respect to $\mathbb{F}_{q}$, but it could be the same element.
	\end{remark}

	\section{1-D Finite Field Fourier Transform and Properties for Constacyclic Codes}\label{section3}
	We begin this section with some basic definitions and concepts related to finite field Fourier transforms.\par

	Suppose $\lambda$ is a nonzero element in $\mathbb{F}_q$ having order $t_1$. Let $\gamma$ be a primitive $M^{\mathrm{th}}$ root of $\lambda$ in an extension field $\mathbb{F}_{q^t}$ for some positive integer $t$. The following lemma gives such an extended field $\mathbb{F}_{q^t}$.
	\begin{lemma}\label{extLemma1D}
		Let $\lambda \in \mathbb{F}^{\ast}_q$ having order $t_1$ and let $\gamma$ be a primitive $M^{\mathrm{th}}$ root of $\lambda$ in an extension field $\mathbb{F}_{q^t}$ for some positive integer $t$. Let $\alpha$ be a generator of the cyclic group $\mathbb{F}_{q^t}^{\ast}$ then $\gamma=\alpha^{r_1}$, where $0\leq r_1 \leq q^t-1$. Suppose that $r_1\vert q^t-1$. Then $r_1$ and $t$ satisfy: $$ r_1 t_1 M +1 = q^t.$$
	\end{lemma}
	\begin{proof}
		The result directly follows from the arguments presented in the proof of Lemma \ref{ExtLemma}.
	\end{proof}
	Let $\lambda \in \mathbb{F}_q^{\ast}$ and $\gamma$ be a primitive $M^{\mathrm{th}}$ root of $\lambda $ and let $\xi$ be a primitive $M^{\mathrm{th}}$ root of unity in $\mathbb{F}_{q^t}$, where $t$ is given by Lemma \ref{extLemma1D}. Then $\xi^M=1$ and for any $0\leq i\leq M-1$ we have $(\xi^i)^M= (\xi^M)^i=1$. Thus, each such $\xi^i$ is also a root of $x^M-1$, and $x^M-1=(x-1)(x-\xi)(x-\xi^2)\cdots(x-\xi^{M-1}).$ From \cite{patel2024quantum}, we have
	\begin{equation}\label{eq1}
		x^M-\lambda=\prod_{j=0}^{M-1}(x-\gamma \xi^j).
	\end{equation}
	Next, we consider $\mathbb{F}_{q^t}$-algebra isomorphism 
	$$\Phi': \frac{\mathbb{F}_{q^t}[x] }{\langle x^{M}-\lambda \rangle} \longrightarrow \prod_{j=0}^{M-1}\frac{\mathbb{F}_{q^t}[x]}{\langle x-\gamma \xi^j\rangle}$$
	defined by 
	\begin{equation}\label{eq2}
		\Phi'\bigg(\sum_{i=0}^{M-1}a_ix^i\bigg)= \bigg(\sum_{i=0}^{M-1}a_i\gamma^i,\sum_{i=0}^{M-1}a_i(\gamma\xi)^i,\dots,\sum_{i=0}^{M-1}a_i(\gamma\xi^{M-1})^i\bigg).
	\end{equation}
	Now, we can define an isomorphism to relate the combinatorial and algebraic structures of constacyclic codes. For this, we consider $\Phi'': \mathbb{F}_{q^t}^M \longrightarrow  \frac{\mathbb{F}_{q^t} [x]}{\langle x^{M}-\lambda \rangle}$ defined by 
	\begin{equation}\label{eq3}
		\Phi''(a_0,a_1,\dots,a_{M-1})=\sum_{i=0}^{M-1}a_ix^i.
	\end{equation}
	From Eqs. \eqref{eq2} and
	\eqref{eq3}, we can define 
	$A_j=\sum_{i=0}^{M-1}a_i(\gamma\xi^j)^i \textnormal{ for } j=0,1,\dots, M-1.$
	
	In order to extend Blahut's approach \cite{Blahut_book} of finite field Fourier transforms for cyclic codes towards the constacyclic case, we begin with some definitions.
	\begin{definition}\label{Def_1}
		Let $\gamma $ be a primitive $M^{\mathrm{th}}$ root of $\lambda \in \mathbb{F}_q^{\ast}$ and let $\xi$ be a primitive $M^{\mathrm{th}}$ root of unity in an extended field, say,  $\mathbb{F}_{q^t}$ for some $t\in \mathbb{N}$. Then 1-D finite field Fourier transform (FFFT) is a function from $\mathbb{F}_q^M \longrightarrow \mathbb{F}_{q^t}^M$ defined by
		\begin{equation}
			{\rm FFFT}(\boldsymbol{a}) =  \boldsymbol{A},
		\end{equation}
		where
		\begin{equation}
			A_j=\sum_{i=0}^{M-1}a_i(\gamma\xi^j)^i ~\text{for }~j=0,1,\dots, M-1.
		\end{equation}
		Here $\boldsymbol{A}$ is called spectrum of $\boldsymbol{a}$.
	\end{definition}
	Next result gives a connection between a constacyclic codeword (in time-domain) and its corresponding spectrum (in frequency/transform-domain).
	\begin{lemma}\label{lemma1}\cite{patel2024quantum}
		Let $\xi$ be an element of $\mathbb{F}_{q^t}$ of order $M$ and $\lambda=\gamma^M$, where  $\gamma \in \mathbb{F}_{q^t}^{\ast}$. A $\lambda$-constacyclic code vector over $\mathbb{F}_q$ and its corresponding spectrum are related by:
		\begin{equation}
			A_j=\sum_{i=0}^{M-1}a_i(\gamma\xi^j)^i, \, \, 
			a_i=\frac{1}{M\gamma^i}\sum_{j=0}^{M-1}\xi^{-ij}A_j.
		\end{equation}
	\end{lemma}
	Now, we recall the following properties of constacyclic codes in the transform-domain. For more details, the readers are referred to \cite{patel2024quantum}.
	\begin{theorem} (Constacyclic shift property in the transform-domain)
		If $\boldsymbol{A} = \text{FFFT}(\boldsymbol{a})$, $\boldsymbol{b}\in\mathbb{F}_q^M$
		such that $b_i = a_{i-1}$ for $i=1,2,\dots,M-1$ and $b_0=\lambda a_{M-1}$, and
		$\boldsymbol{B} = \text{FFFT} (\boldsymbol{b})$, then $B_j =  \gamma\xi^j A_j$.
	\end{theorem}
	\begin{theorem} (Convolutional Property)
		If $\boldsymbol{a},\boldsymbol{b}$ and $\boldsymbol{c}$ are $M$-tuple vectors over $\mathbb{F}_q$ in the time-domain such that $c_i=a_ib_i$ for $i=0,1,\dots,M-1$, then their FFFT coefficients satisfy the relation 
		$$C_j= \frac{1}{M}\sum_{k=0}^{M-1}B_k\mathcal{A}_{j-k},$$
		where $j=0,1,\dots,M-1$ and $\mathcal{A}_{j-k}$ denotes the spectrum of finite field Fourier transform of cyclic codes (in this case $\gamma =1$), and conversely.
		
	\end{theorem}
	\begin{theorem} (Conjugate Symmetry Property)\label{th3}
		Let $A_j$ for $j=0,1,\dots,n-1$ take elements in $\mathbb{F}_{q^m}$ and $n|(q^m-1)$. Then, $a_i=0$ for $i=0,1,\dots,n-1$ are all elements of  $\mathbb{F}_{q}$ if and only if the following
		equations are satisfied:
		$A_j^q=A_{qj\mod(n)}$, $j=0,1,\dots,n-1$. 
	\end{theorem}
	\begin{theorem} (Reversal Preserving Property)
		Let $\boldsymbol{a}=(a_0,a_1,\dots,a_{n-1})$ and  $\boldsymbol{b}=(b_0,b_1,\dots,b_{n-1})$ be two vectors over $\mathbb{F}_q$. Let $\boldsymbol{A}=(A_0,A_1,\dots,A_{n-1})$ and  $\boldsymbol{B}=(B_0,B_1,\dots,B_{n-1})$ be their transform vectors, and $b_i=a_{n-1-i}$ for all $i=0,1,\dots,n-1$. Then, $B_j=A_{n-1-j}$ for all $j=0,1,\dots,n-1$. 
	\end{theorem}

	\section{Characterization of 2-D Constacyclic Codes in the Transform domain}\label{section4}
	
	We next focus on transform-domain properties in the 2-D setup, and characterization of 2-D constacyclic codes over $\mathbb{F}_q$.
	\subsection{Two-dimensional finite field Fourier transform}
	Two-dimensional finite field Fourier transform (2-D FFFT) \cite{Blahut_book} is an extension of the 1-D finite field Fourier transform (1-D FFFT). Instead of a row vector, we consider an array of size $M\times N$ over $\mathbb{F}_q$. Consider a 2-D array of size $M\times N$ over $\mathbb{F}_q$ as  
	\begin{equation} c=\begin{pmatrix}
			c_{0,0} & c_{0,1}  & \dots & c_{0,N-1}\\
			c_{1,0} & c_{1,1}  & \dots & c_{1,N-1}\\
			\vdots & \vdots & \ddots & \vdots\\
			c_{M-1,0} & c_{M-1,1} & \dots & c_{M-1,N-1}
		\end{pmatrix}.
	\end{equation}
	
	From Definition \ref{Def_1}, we define the 2-D FFFT as follows.	Let $\gamma, \beta$ be primitive $M^{\rm th}$ and $N^{\rm th}$ roots of $\lambda_1, \lambda_2 \in \mathbb{F}_q$, respectively, and let $\zeta_1, \zeta_2$ be $M^{\rm th}$ and $N^{\rm th}$ roots of unity, respectively, in an extended field, say $\mathbb{F}_{q^t}$, that is, $\zeta_1^{M}=1$, $\zeta_2^{N}=1$, and $\gamma^{M}=\lambda_1$, $\beta^{N}=\lambda_2$.
	Then 2-D FFFT is a mapping from $\mathbb{M}_{M\times N}(\mathbb{F}_q) \longrightarrow \mathbb{M}_{M\times N}(\mathbb{F}_{q^t})$ defined by
	\begin{equation}\label{2DFFFT}
		{\rm FFFT}(c)= C,
	\end{equation}
	where \begin{equation} C=\begin{pmatrix}
			C_{0,0} & C_{0,1}  & \dots & C_{0,N-1}\\
			C_{1,0} & c_{1,1}  & \dots & C_{1,N-1}\\
			\vdots & \vdots & \ddots & \vdots\\
			C_{M-1,0} & C_{M-1,1} & \dots & C_{M-1,N-1}
		\end{pmatrix}
	\end{equation}
	and 
	\begin{equation}\label{vectorFreq}
		C_{\theta,\phi}=\sum_{i=0}^{M-1}\sum_{j=0}^{N-1}c_{i,j}{(\gamma{\zeta_1}^{\theta})}^i{(\beta{\zeta_2}^{\phi})}^j
	\end{equation}
	for $(\theta, \phi)\in \Omega$.
	We may replace the double summation ${\displaystyle \sum_{i=0}^{M-1}}{\displaystyle \sum_{j=0}^{N-1}}$ by ${\displaystyle \sum_{(i,j)\in\Omega}}$, where
	\begin{equation}
		\Omega=\{(i,j)|0\leq i\leq M-1 \, \textnormal{ and } 0\leq j\leq N-1\}.
	\end{equation}
	\begin{remark}
		Note that FFFT defined by Eq. \eqref{2DFFFT} is an injective function, but not surjective.
	\end{remark}
	Hence, the inverse finite field Fourier transform (2-D IFFFT) of an element in the range of Eq. \eqref{2DFFFT} is defined as 
	\begin{equation}
		{\rm IFFFT}(C)=c,
	\end{equation}
	where 
	\begin{equation}\label{vectorTime}
		c_{i,j}=\frac{1}{M(\text{mod }p)\cdot N(\text{mod }p)}\sum_{\theta=0}^{M-1}\sum_{\phi=0}^{N-1}C_{\theta,\phi}{(\gamma{\zeta_1}^{\theta})}^{-i}{(\beta{\zeta_2}^{\phi})}^{-j} \ \text{ for } (i, j)\in \Omega.
	\end{equation}
	The following result establishes a connection between 2-D constacyclic codes in the time-domain and its corresponding FFFT images in the frequency-domain.
	\begin{lemma}\label{lem1}
		Let $\gamma, \beta$ be primitive $M^{\rm th}$ and $N^{\rm th}$ roots of $\lambda_1, \lambda_2 \in \mathbb{F}_q$, respectively, and let $\zeta_1, \zeta_2$ be  $M^{\rm th}$ and $N^{\rm th}$ roots of unity, respectively, in an extended field, say, $\mathbb{F}_{q^t}$. Let $\mathcal{C}$ be 2-D $(\lambda_1, \lambda_2)$-constacyclic code  over $\mathbb{F}_q$. Then a codeword in $\mathcal{C}$ and its corresponding spectrum array over $\mathbb{F}_{q^t}$ are related by: 
		\begin{equation}\label{main-C-c}
			\begin{split}
				C_{\theta, \phi}=&\sum_{i=0}^{M-1}\sum_{j=0}^{N-1}c_{i,j}(\gamma\zeta_1^{\theta})^{i}(\beta \zeta_2^{\phi})^{j}\ \text{ for } (\theta, \phi)\in \Omega,\\ 
				c_{i,j}=&\frac{1}{M(\text{\rm mod }p)\cdot N(\text{\rm mod }p)}\sum_{\theta=0}^{M-1}\sum_{\phi=0}^{N-1}C_{\theta, \phi}(\gamma \zeta_1^{\theta})^{-i}(\beta \zeta_2^{\phi})^{-j} \ \text{ for } (i, j)\in \Omega.
			\end{split}
		\end{equation}
	\end{lemma}
	\begin{proof}
		From Eq. \eqref{eq1}, we rewrite  
		\begin{equation}\label{2DXpoly}
			\begin{split}
				x^M-\gamma^M=&(x-\gamma)(x^{M-1}+\gamma x^{M-2}+\cdots+\gamma^{M-2}x+\gamma^{M-1})\\
				=&(x-\gamma)\bigg(\sum_{i=0}^{M-1}\gamma^ix^{M-1-i}\bigg).
			\end{split}
		\end{equation}
		From the definition of $\zeta_1$, $\gamma\zeta_1^r$ is a root of the above polynomial, where $r=0,1,...,M-1$. Hence, $\gamma\zeta_1^r$ is a root of the last term in RHS of Eq. \eqref{2DXpoly} for all $r \neq 0 \, (\text{mod }M)$. This can also be written as 
		\begin{equation}\label{xyzEq1}
			\sum_{i=0}^{M-1}\gamma^i(\gamma\zeta_1^{r})^{M-1-i}=(\gamma\zeta_1^r)^{M-1}\sum_{i=0}^{M-1}\zeta_1^{-ri}=0, \textnormal{ for } r \neq 0\, (\text{mod } M).
		\end{equation} 
		If $r=0$, then 
		\begin{equation}
			\sum_{i=0}^{M-1}(\gamma\zeta_1^r)^{M-1}\zeta_1^{-ri}=\gamma^{M-1} M \, (\text{mod }p).
		\end{equation} 
		This is not zero if $M$ is not a multiple of the field characteristic $p$.\\
		Similarly, we have  
		\begin{equation}\label{2DYpoly}
			\begin{split}
				y^N-\beta^N=(y-\beta)\bigg(\sum_{i=0}^{N-1}\beta^iy^{N-1-i}\bigg).
			\end{split}
		\end{equation}
		From the definition of $\zeta_2$, $\beta\zeta_2^r$ is a root of the above polynomial, where $r=0,1,...,N-1$. Hence, $\beta\zeta_2^r$ is a root of the last term in RHS of Eq. \eqref{2DYpoly} for all $r \neq 0\  (\text{mod } N) $. This can also be written as 
		\begin{equation}\label{xyzEq2}
			\sum_{i=0}^{N-1}\beta^i(\beta\zeta_2^{r})^{N-1-i}=(\beta\zeta_2^r)^{N-1}\sum_{i=0}^{N-1}\zeta_2^{-ri}=0, \textnormal{ for } r \neq 0 \, (\text{mod } N).
		\end{equation} 
		If $r=0$, then 
		\begin{equation}
			\sum_{i=0}^{N-1}(\beta\zeta_2^r)^{N-1}\zeta_2^{-ri}=\beta^{N-1} N\, (\text{mod } p).
		\end{equation} 
		This is not zero if $N$ is not a multiple of the field characteristic $p$.\\
		Now, we have
		\begin{equation}\label{xyzEq3}
			\begin{split}
				\sum_{\theta=0}^{M-1}\sum_{\phi=0}^{N-1}C_{\theta, \phi}(\gamma \zeta_1^{\theta})^{-i}(\beta \zeta_2^{\phi})^{-j}
				=&\sum_{\theta=0}^{M-1}\sum_{\phi=0}^{N-1}(\gamma \zeta_1^{\theta})^{-i}(\beta \zeta_2^{\phi})^{-j} \sum_{k=0}^{M-1}\sum_{k'=0}^{N-1}c_{k,k'}(\gamma\zeta_1^{\theta})^{k}(\beta \zeta_2^{\phi})^{k'}\\
				=& \sum_{k=0}^{M-1}\sum_{k'=0}^{N-1}c_{k,k'} \sum_{\theta=0}^{M-1}\sum_{\phi=0}^{N-1}(\gamma \zeta_1^{\theta})^{(k-i)}(\beta \zeta_2^{\phi})^{(k'-j)}\\
				=& \sum_{k=0}^{M-1}\sum_{k'=0}^{N-1}c_{k,k'} \gamma^{k-i} \beta^{k'-j}\Big(\sum_{\theta=0}^{M-1} \zeta_1^{\theta(k-i)}\Big) \Big(\sum_{\phi=0}^{N-1} \zeta_2^{\phi(k'-j)}\Big).\\
			\end{split}
		\end{equation}
		Using Eqs. \eqref{xyzEq1}, \eqref{xyzEq2} in Eq. \eqref{xyzEq3}, we have 
		\begin{equation}
			\begin{split}
				\sum_{\theta=0}^{M-1}\sum_{\phi=0}^{N-1}C_{\theta, \phi}(\gamma \zeta_1^{\theta})^{-i}(\beta \zeta_2^{\phi})^{-j}=c_{i,j}\, M \, (\text{mod}\, {p})\cdot N (\text{mod}\, {p}).
			\end{split}
		\end{equation}
		Since $M, N$ are not multiple of $p$, $M \, (\text{mod}\, {p}) \neq 0, N\, (\text{mod}\, {p})\neq 0$, proving the lemma.
	\end{proof}
	
	\begin{example}
		Let $\lambda_1=\lambda_2=2 \in \mathbb{F}_3$. Consider a 2-D $(2, 2)$-constacyclic codeword of size $4\times 5$ over $\mathbb{F}_3$
		\[
		c=
		\left[\begin{array}{ccccc}
			2 & 0 & 2 & 0 & 0\\
			1 & 2 & 0 & 0 & 0\\
			0 & 0 & 0 & 0 & 0\\
			0 & 0 & 0 & 0 & 0\\
		\end{array}\right].
		\]
		Suppose $\alpha$ is a root of a primitive polynomial $x^4 + 2x^3 + 2$, that is, $\mathbb{F}_3(\alpha)=\mathbb{F}_{3^4}$. With $\gamma = \alpha^{10}$ and $\beta = \alpha^{8}$ as the $4^{\mathrm{th}}$ and $5^{\mathrm{th}}$ roots of $\lambda_1$ and $\lambda_2$, respectively, the transformed elements belong to $\mathbb{F}_{3^4}$. Let $\zeta_1 = \alpha^{20},\zeta_2=\alpha^{16} \in \mathbb{F}_{3^4}$ be primitive $4^{\mathrm{th}}$ and $5^{\mathrm{th}}$ roots of unity. With this setting, we are ready to perform the 2-D FFFT operation on
		$c$. By using Eq. \eqref{vectorFreq}, the 2-D FFFT of $c$ is
		\[
		{\rm FFFT}(c)=C=\left[\begin{array}{ccccc}
			\alpha^{33} & \alpha^{17} & \alpha^{30} & \alpha^{73} & \alpha^{57}\\
			\alpha^{59} & \alpha^{19} & \alpha^{10} & \alpha^{11} & \alpha^{51}\\
			\alpha^{12} & \alpha^{50} & \alpha^{20} & \alpha^{50} & \alpha^{28}\\
			\alpha^{70} & \alpha^{36} & \alpha^{60} & \alpha^{4} & \alpha^{70}\\
		\end{array}\right].
		\]
		One can use Eq. \eqref{vectorTime}, to calculate 2-D IFFFT of $C$ to get the array $c$.
	\end{example}
	
	The following result extends Lemma \ref{extLemma1D} and
	gives the condition on the smallest extension field that the transformed codewords must lie.  
	
	\begin{lemma}\cite{CCDS}\label{ExtLemma}
		Suppose $\lambda_1, \lambda_2$ are two nonzero elements in $\mathbb{F}_q$ having order $t_1$ and $t_2$, respectively. Let $\gamma, \beta$ be primitive $M^{\mathrm{th}}, N^{\mathrm{th}}$ roots of $\lambda_1, \lambda_2$ respectively, in an extension field $\mathbb{F}_{q^t}$ for some positive integer $t$. Let $\alpha$ be a generator of the cyclic group $\mathbb{F}_{q^t}^{\ast}$ then $\gamma=\alpha^{r_1}$, $\beta=\alpha^{r_2}$ for some $r_1, r_2$, where $0\leq r_1, r_2 \leq q^t-2$. Then $M, N, t_i, r_i$ and $t$ satisfy the following conditions:
		\begin{enumerate}
			\item $ r_1 t_1 M +1 = q^t$ and 
			\item $ r_2 t_2 N +1 = q^t$.
		\end{enumerate}
	\end{lemma}
	
	The properties of 2-D FFFT \cite{Blahut_book} are the same as those of 1-D FFFT, as shown before. We will discuss some of the relevant properties of 2-D FFFT applicable to 2-D constacyclic codes. We refer the readers to \cite{patel2024quantum} for more details.
	
	\subsection{Properties of two-dimensional finite field Fourier transform}
	In this section, we derive some important properties of the 2-D finite field Fourier transform.
	
	\begin{theorem} Let $\textbf{a}, \textbf{b}\in \mathbb{M}_{M\times N}(\mathbb{F}_q)$ and let $\textbf{A}=\text{FFFT}(\textbf{a}), \textbf{B}=\text{FFFT}(\textbf{b}) \in \mathbb{M}_{M\times N}(\mathbb{F}_{q^t})$.
		\begin{enumerate}
			\item (Column $\lambda_1$-constacyclic shift property in the transform domain) If $\textbf{b}$ is the column $\lambda_1$-constacyclic shift of $\textbf{a}$, i.e., $b_{0,j}= \lambda_1 a_{M-1, j}$ and $b_{i,j}=a_{i-1, j}$ for $1\leq i \leq M-1$, \, $0\leq j \leq N-1$, then $B_{\theta, \phi} = \gamma \zeta_1^{\theta} A_{\theta, \phi}$.
			
			\item (Row $\lambda_2$-constacyclic shift property in the transform domain) If $\textbf{b}$ is the row $\lambda_2$-constacyclic shift of $\textbf{a}$, that is, $b_{i,0}= \lambda_2 a_{i, N-1}$ and $b_{i,j}=a_{i, j-1}$ for $0\leq i \leq M-1$, \, $1\leq j \leq N-1$, then $B_{\theta, \phi} = \beta \zeta_2^{\phi} A_{\theta, \phi}$.
		\end{enumerate}
		\begin{proof}
			Let $\textbf{A}=\text{FFFT}(\textbf{a}), \textbf{B}=\text{FFFT}(\textbf{b})$. Then
			\begin{equation*}
				\begin{split}
					B_{\theta, \phi} &= \sum_{i=0}^{M-1} \sum_{j=0}^{N-1} b_{i,j}(\gamma \zeta_1^{\theta})^{i} (\beta \zeta_2^{\phi})^{j}\\
					&=\sum_{j=0}^{N-1} b_{0,j}(\gamma \zeta_1^{\theta})^{0} (\beta \zeta_2^{\phi})^{j} + \sum_{i=1}^{M-1} \sum_{j=0}^{N-1} b_{i,j}(\gamma \zeta_1^{\theta})^{i} (\beta \zeta_2^{\phi})^{j}\\
					&=\sum_{j=0}^{N-1} \lambda_1 a_{M-1,j} (\beta \zeta_2^{\phi})^{j} + \sum_{i=1}^{M-1} \sum_{j=0}^{N-1} a_{i-1,j}(\gamma \zeta_1^{\theta})^{i} (\beta \zeta_2^{\phi})^{j}\\
					&=\sum_{j=0}^{N-1} a_{M-1,j}(\lambda_1 \zeta_1^{\theta M}) (\beta \zeta_2^{\phi})^{j} + \sum_{i=1}^{M-1} \sum_{j=0}^{N-1} a_{i-1,j}(\gamma \zeta_1^{\theta})^{i} (\beta \zeta_2^{\phi})^{j}\\
					&=\sum_{j=0}^{N-1} a_{M-1,j} (\gamma \zeta_1^{\theta})^{M}(\beta \zeta_2^{\phi})^{j} + \sum_{i=1}^{M-1} \sum_{j=0}^{N-1} a_{i-1,j}(\gamma \zeta_1^{\theta})^{i} (\beta \zeta_2^{\phi})^{j}\\
					&=\sum_{i=0}^{M-1} \sum_{j=0}^{N-1} a_{i,j}(\gamma \zeta_1^{\theta})^{i+1} (\beta \zeta_2^{\phi})^{j}\\
					&=\gamma \zeta_1^{\theta}A_{\theta, \phi}.
				\end{split}
			\end{equation*}
			One can prove part (2) similarly.
		\end{proof}
	\end{theorem}
	\begin{theorem} (Convolutional Property)
		Let $\textbf{a}, \textbf{b}, \textbf{c} \in \mathbb{M}_{M\times N}(\mathbb{F}_q)$ and $\textbf{A}, \textbf{B}$, and $\textbf{C}$ be their corresponding FFFT, respectively.
		\begin{enumerate}
			\item (Multiplication to convolution) Then $c_{i,j}=a_{i,j}b_{i,j}$ for all $(i, j)\in \Omega$ if and only if 
			\begin{equation*}
				C_{\theta, \phi}=\frac{1}{M (\textnormal{ mod } p)\cdot  N (\textnormal{ mod } p)} \sum_{\theta'=0}^{M-1} \sum_{\phi'=0}^{N-1} \mathcal{A}_{\theta -\theta', \phi - \phi'} B_{\theta', \phi'}
			\end{equation*}
			for all $(\theta, \phi)\in \Omega$, where $\mathcal{A}_{\theta -\theta', \phi - \phi'}$ denotes the finite field Fourier transform of cyclic codes (that is, when $\gamma =1, \beta =1$).
			
			\item (Convolution to multiplication)  Then $C_{\theta, \phi}=A_{\theta, \phi}B_{\theta, \phi}$ for all $(\theta, \phi)\in \Omega$ if and only if 
			\begin{equation*}
				c_{i, j}= \sum_{i'=0}^{M-1}\sum_{j'=0}^{N-1} a_{i -i', j - j'} b_{i', j'}
			\end{equation*}
			for all $(i, j)\in \Omega$.
		\end{enumerate}
		\begin{proof}
			Let $c_{i,j}=a_{i,j}b_{i,j}$ for all $(i,j)\in \Omega$. Then
			\begin{equation*}
				\begin{split}
					C_{\theta, \phi} &= \sum_{i=0}^{M-1} \sum_{j=0}^{N-1} c_{i,j}(\gamma \zeta_1^{\theta})^{i} (\beta \zeta_2^{\phi})^{j}\\
					&=\sum_{i=0}^{M-1} \sum_{j=0}^{N-1} a_{i,j}b_{i,j}(\gamma \zeta_1^{\theta})^{i} (\beta \zeta_2^{\phi})^{j}\\
					&=\sum_{i=0}^{M-1} \sum_{j=0}^{N-1} a_{i,j}(\gamma \zeta_1^{\theta})^{i} (\beta \zeta_2^{\phi})^{j}\cdot \frac{1}{M (\textnormal{ mod } p)\cdot N (\textnormal{ mod } p)} \sum_{\theta'=0}^{M-1} \sum_{\phi'=0}^{N-1}B_{\theta', \phi'}(\gamma \zeta_1^{\theta'})^{-i} (\beta \zeta_2^{\phi'})^{-j}\\
					&=\frac{1}{M (\textnormal{ mod } p)\cdot N (\textnormal{ mod } p)}\sum_{\theta'=0}^{M-1} \sum_{\phi'=0}^{N-1} \Big(\sum_{i=0}^{M-1} \sum_{j=0}^{N-1} a_{i, j} \zeta_1^{i(\theta -\theta')}  \zeta_2^{j(\phi -\phi')}\Big)B_{\theta', \phi'}\\
					&=\frac{1}{M (\textnormal{ mod } p)\cdot N (\textnormal{ mod } p)}\sum_{\theta'=0}^{M-1} \sum_{\phi'=0}^{N-1} \mathcal{A}_{\theta -\theta', \phi -\phi'}B_{\theta', \phi'}.
				\end{split}
			\end{equation*}
			The converse follows from the uniqueness of FFFT.\\
			Part (2) can be proved similarly.
		\end{proof}
	\end{theorem}

	\begin{theorem} (Conjugate Symmetry Property)\label{th3}
		Let $\textbf{C} \in \mathbb{M}_{M\times N}(\mathbb{F}_{q^m})$ and $\textbf{c} \in \mathbb{M}_{M\times N}(\mathbb{F}_{q})$. Then $\textbf{C}$ is the finite field Fourier transform of $\textbf{c}$ if and only if
		\begin{equation}
			C_{\theta, \phi}^{q}= \sum_{i=0}^{M-1} \sum_{j=0}^{N-1}c_{i,j}(\gamma \zeta_1^{\theta})^{iq} (\beta \zeta_2^{\phi})^{jq}
		\end{equation}
		for all $(\theta, \phi)\in \Omega$.
		\begin{proof}
			Let $\textbf{C}$ be the FFFT of $\textbf{c}$. Then by the definition of FFFT, we have
			\begin{equation}\label{thisEq}
				C_{\theta, \phi}= \sum_{i=0}^{M-1} \sum_{j=0}^{N-1}c_{i,j}(\gamma \zeta_1^{\theta})^{i} (\beta \zeta_2^{\phi})^{j}.
			\end{equation}
			We know that $(a+b)^q=a^q+b^q$ for all $a, b \in \mathbb{F}_{q^m}$, and $d^q=d$ for all $d\in \mathbb{F}_q$. From Eq. \eqref{thisEq}, we have
			\begin{equation*}
				\begin{split}
					C_{\theta, \phi}^{q}=& \Big(\sum_{i=0}^{M-1} \sum_{j=0}^{N-1}c_{i,j}(\gamma \zeta_1^{\theta})^{i} (\beta \zeta_2^{\phi})^{j}\Big)^q\\
					=& \sum_{i=0}^{M-1} \sum_{j=0}^{N-1}c_{i,j}^q (\gamma \zeta_1^{\theta})^{iq} (\beta \zeta_2^{\phi})^{jq}\\
					=&  \sum_{i=0}^{M-1} \sum_{j=0}^{N-1}c_{i,j}(\gamma \zeta_1^{\theta})^{iq} (\beta \zeta_2^{\phi})^{jq}.
				\end{split}
			\end{equation*}
			Conversely, for all $(\theta, \phi)\in \Omega$, let
			$$C_{\theta, \phi}^{q}=\sum_{i=0}^{M-1} \sum_{j=0}^{N-1}c_{i,j}(\gamma \zeta_1^{\theta})^{iq} (\beta \zeta_2^{\phi})^{jq}.$$\\
			We have
			\begin{equation*}
				\begin{split}
					\sum_{i=0}^{M-1} \sum_{j=0}^{N-1}c_{i,j}^q (\gamma \zeta_1^{\theta})^{iq} (\beta \zeta_2^{\phi})^{jq}
					=&  \sum_{i=0}^{M-1} \sum_{j=0}^{N-1}c_{i,j}(\gamma \zeta_1^{\theta})^{iq} (\beta \zeta_2^{\phi})^{jq}.
				\end{split}
			\end{equation*}
			By the uniqueness of the FFFT, we have $c_{i,j}^q=c_{i,j}$ for all $(i,j)\in \Omega$. Thus each $c_{i,j}$ is a root of $x^q-x=0 \implies c_{i,j}\in \mathbb{F}_q$.
		\end{proof}
	\end{theorem}

	\subsection{Spectral nulls of a 2-D $(\lambda_1, \lambda_2)$-constacyclic code}
	We will derive the relationship between the roots of an array in the time-domain (respectively, in frequency domain) and its corresponding array in the frequency domain (respectively, in time-domain).
	\begin{definition}
		For a given 2-D linear code, the locations $(\theta, \phi)$ at which the frequency domain components are zero, i.e., $C_{\theta, \phi}=0$, we refer to them as the spectral nulls.
	\end{definition}
	Consider a 2-D array $c$ of area $M\times N$ over $\mathbb{F}_q$. Then its polynomial representation is 
	\[
	c(x,y)=\sum_{i=0}^{M-1}\sum_{j=0}^{N-1}c_{i,j}x^{i}y^{j}.
	\]
	
	Now, the root of this polynomial at $\left(\gamma \zeta_1^{\theta},\beta \zeta_2^{\phi}\right)$
	is 
	\[
	c(\gamma \zeta_1^{\theta},\beta \zeta_2^{\phi})=\sum_{i=0}^{M-1}\sum_{j=0}^{N-1}c_{i,j}(\gamma \zeta_1^{\theta})^i (\beta \zeta_2^{\phi})^j.
	\]
	
	From Eq. \eqref{vectorFreq}, we have 
	\begin{equation}\label{spectralnullEq}
		c(\gamma \zeta_1^{\theta},\beta \zeta_2^{\phi})= C_{\theta,\phi}.
	\end{equation}
	
	Thus, the frequency domain components at the coordinate $(\theta,\phi)$ are values of the 2-D (or bi-variate) polynomial in the time domain at $(\gamma \zeta_1^{\theta},\beta \zeta_2^{\phi})$.\par 
	Now, we have the following result.
	\begin{theorem}\label{th6}
		Let $c(x,y)$ be a polynomial in the time-domain and its spectrum polynomial ${\rm FFFT}(c(x,y)) = C(x, y)$. Then we have the following:
		\begin{enumerate}
			\item The polynomial $c(x, y)$ has a root at $(\gamma \zeta_1^{\theta}, \beta \zeta_2^{\phi})$ if and only if $C_{\theta, \phi}$ equals zero.
			\item  The polynomial $C(x, y)$ has a root at $ (\zeta_1^{-i}, \zeta_2^{-j})$ if and only if $c_{i, j}$ equals zero.
		\end{enumerate}    
		\begin{proof}
			From Eq. \eqref{main-C-c}, we have the following connection between the coordinates of an array and their corresponding spectrum 
			\begin{equation}
				\begin{split}
					C_{\theta, \phi}=&\sum_{i=0}^{M-1}\sum_{j=0}^{N-1}c_{i,j}(\gamma\zeta_1^{\theta})^{i}(\beta \zeta_2^{\phi})^{j}, \\ 
					c_{i,j}=&\frac{1}{M(\text{\rm mod }p)\cdot N(\text{\rm mod }p)}\sum_{\theta=0}^{M-1}\sum_{\phi=0}^{N-1}C_{\theta, \phi}(\gamma \zeta_1^{\theta})^{-i}(\beta \zeta_2^{\phi})^{-j}.
				\end{split}
			\end{equation}
			\begin{enumerate}
				\item Since $c(x, y) = \sum_{i=0}^{M-1}\sum_{j=0}^{N-1}c_{i, j}x^i y^j$, we have 
				\begin{align}\label{eq7}
					c(\gamma \zeta_1^{\theta}, \beta \zeta_2^{\phi})=\sum_{i=0}^{M-1}\sum_{j=0}^{N-1} c_{i, j}(\gamma \zeta_1^{\theta})^i (\beta \zeta_2^{\phi})^j=C_{\theta, \phi}.
				\end{align} 
				This proves the required result.
				\item Since $C(x, y)=\sum_{\theta=0}^{M-1} \sum_{\phi=0}^{N-1} C_{\theta, \phi}x^{\theta}y^{\phi}$, we have
				$$ C( \zeta_1^{-i}, \zeta_2^{-j})=\sum_{\theta=0}^{M-1}\sum_{\phi=0}^{N-1}C_{\theta, \phi}(\zeta_1^{-i})^{\theta} (\zeta_2^{-j})^{\phi} = M N\cdot  \gamma^i \beta^j c_{i, j},$$ proving the required result.
			\end{enumerate}
		\end{proof}
	\end{theorem}
	
	For a given 2-D $(\lambda_1, \lambda_2)$-constacyclic code with its CZ set $V_c$, we have the following result. We refer the readers to \cite{CCDS} for more details.
	\begin{lemma}\label{spectralNullLemma}
		Let $\mathcal{C}$ be a 2-D $(\lambda_1, \lambda_2)$-constacyclic code with its CZ set $V_c$. If $(\gamma \zeta_1^{\theta},\beta \zeta_2^{\phi})\in V_c$, then the value at the location $(\theta,\phi)$  of the spectrum array $C$ in the frequency domain is zero, that is, $C_{\theta,\phi}=0$ for each codeword $c\in \mathcal{C}$.
		\begin{proof}
			Let $\mathcal{C}$ be a 2-D $(\lambda_1, \lambda_2)$-constacyclic code with its CZ set $V_c$. Then by the definition of the CZ set, we have
			\begin{equation*}
				(\gamma \zeta_1^{\theta},\beta \zeta_2^{\phi})\in V_c \implies c(\gamma \zeta_1^{\theta},\beta \zeta_2^{\phi}) = 0 \ \forall \ c\in C.
			\end{equation*}
			By using Eq. \eqref{spectralnullEq}, we have $C_{\theta, \phi}=0$.
		\end{proof}
	\end{lemma}
	\begin{remark}
		Note that we have $\vert V_c \vert = \vert \Pi \vert $ many spectral nulls in each spectrum of codewords in the frequency domain, for a given 2-D $(\lambda_1, \lambda_2)$-constacyclic code with CZ set $V_c$.
	\end{remark}
	The following example illustrates the above result.
	\begin{example}\label{NullExamp}
		Suppose $\lambda_1 = \lambda_2 =2 \in \mathbb{F}_3$. Let us fix $\gamma$ and $\beta$ as primitive  $4^{\mathrm{th}}$ and $5^{\mathrm{th}}$ roots of $\lambda_1$ and $\lambda_2$, respectively. Let $\gamma=\alpha^{10}$, $\beta=\alpha^{8}$, where $\alpha$ is a root of a primitive polynomial $x^4+x+2$, that is, $\mathbb{F}_{3^4}=\mathbb{F}_{3}(\alpha)$. Here, the conjugate set of $\gamma = \{\gamma, \gamma^3\}$, and the conjugate set of $\beta = \{\beta, \beta^3, 2\beta^4, 2\beta^2\}$.
		Let $\mathcal{C}$ be a $(2,2)$-constacyclic code of area $4\times 5$ over $\mathbb{F}_{3}$ having the following CZ set
		\begin{equation}\label{EqforspecNullEx}
			V_{c}=\{(\gamma, \beta), (\gamma^3, \beta^3), (\gamma, 2\beta^4), (\gamma^3, 2\beta^2); (\gamma, \beta^3), (\gamma^3, 2\beta^4), (\gamma, 2\beta^2), (\gamma^3, \beta)\}
		\end{equation} 
		and ECZ set
		\begin{equation*}
			\hat{V}_{c}=\{(\gamma, \beta), (\gamma, \beta^3)\}.
		\end{equation*} 
		Let $\xi_1=\gamma$, $\eta_{1,1}=\beta$, $\eta_{1,2}=\beta^3$. We have $$g_{\xi 1}(x)=(x-\gamma)(x-\gamma^3)=x^2 + x + 2,$$ the minimal polynomial of $\gamma$ over $\mathbb{F}_3$. \\
		Let $\theta = \alpha^{10}$. Consider $\mathbb{F}_{3^2}=\mathbb{F}_3(\theta)$ subfield of $\mathbb{F}_{3^4}$, where $\theta^2+2\theta+2=0$. Now, we have 
		\begin{equation*}
			\begin{split}
				M_{1,1}(y)&=(y-\beta)(y-\beta^9)\\
				&=(y-\beta)(y-2\beta^4)\\
				&=y^2-(\beta+2\beta^4)y+2\beta^5\\
				&=y^2+\theta y+1,
			\end{split}
		\end{equation*} 
		the minimal polynomial of $\eta_{1,1}$ over $\mathbb{F}_{3^2}$ and 
		\begin{equation*}
			\begin{split}
				M_{1,2}(y)
				&=(y-\beta^3)(y-\beta^{27})\\
				&=(y-\beta^3)(y-2\beta^2)\\
				&=y^2-(\beta^3 + 2\beta^2)y+2\beta^5\\
				&=y^2+\theta^3 y+1,
			\end{split}
		\end{equation*} 
		the minimal polynomial of $\eta_{1,2}$ over $\mathbb{F}_{3^2}$. Thus, $$G_{\eta 1}(y)=M_{1,1}(y)\cdot M_{1,2}(y)=y^4+2y^3+y^2+2y+1.$$ Here, we have $$m_1=2, n_{1,1}=2, n_{1,2}=2, n_{1}=4.$$ 
		Consider the ordered basis $\{1, \alpha, \alpha^2, \alpha^3 \}$ of $\mathbb{F}_{3^4}$ over $\mathbb{F}_{3}$. The check tensor of $\mathcal{C}$ is given by $$\textbf{H}=[\textbf{h}_{k,l}],$$ 
		where $\textbf{h}_{k,l}=\big(\textbf{h}_{1,1}^{(k,l)}, \textbf{h}_{1,2}^{(k,l)}\big)$, and $\textbf{h}_{i,j}^{(k,l)}$ is an $m_{i}n_{i,j}$-tuple coefficient vector of $\xi_{i}^k\eta_{i,j}^l$ over $\mathbb{F}_3$. For instance $\textbf{h}_{1,3}=(\textbf{h}_{1,1}^{(1,3)}, \textbf{h}_{1,2}^{(1,3)})$. Here, $\textbf{h}_{1,1}^{(1,3)}$ is $4$-tuple coefficient vector of $\gamma \beta^3=\alpha^{34}$ over $\mathbb{F}_3$, that is, $\textbf{h}_{1,1}^{(1,3)}=(2210)$. Similarly, we have $\textbf{h}_{1,2}^{(1,3)}=(0010)$. Hence, $\textbf{h}_{1,3}=(22100010)$. By similar calculations, we have

		\[\textbf{H}=
		\begin{pmatrix}
			\textcolor{red}{(1 0 0 0 1 0 0 0)}&
			\textcolor{red}{(1 1 1 0 1 0 1 0)}&
			\textcolor{red}{(2 1 0 2 2 2 2 0)}&
			\textcolor{red}{(1 0 1 0 1 0 1 2)}&
			(2 0 2 1 2 1 0 2)\\
			\textcolor{red}{(1 2 1 1 1 2 1 1)}&
			\textcolor{red}{(0 2 0 1 2 2 1 0)}&
			\textcolor{red}{(1 2 1 0 0 1 0 2)}&
			\textcolor{red}{(2 2 1 0 0 0 1 0)}&
			(0 0 2 0 1 2 1 0)\\
			(0 1 2 2 0 1 2 2)&
			(1 2 1 2 2 1 0 0)&
			(1 2 2 2 2 1 2 1)&
			(2 1 0 0 1 0 0 2)&
			(2 0 0 1 1 2 2 2)\\
			(1 1 2 2 1 1 2 2)&
			(2 0 2 2 0 1 1 0)&
			(0 0 2 1 1 0 1 1)&
			(0 1 1 0 2 0 1 1)&
			(1 0 2 2 0 0 2 1)
		\end{pmatrix}.\]
		From Eq. \eqref{EqforspecNullEx}, for each $(\theta, \phi)\in \{(0,0), (0,1), (0,3), (0,4), (1,0), (1,1), (1,3), (1,4)\}$, we have 
		$$(\gamma \zeta_1^{\theta}, \beta \zeta_2^{\phi}) \in V_c .$$ 
		By using Lemma \ref{spectralNullLemma}, we get $C_{\theta, \phi}=0$\, $\forall \, (\theta, \phi)\in \{(0,0), (0,1), (0,3), (0,4), (1,0), (1,1), (1,3), (1,4)\}$. The corresp- onding time and frequency domain codewords are displayed in Figure \ref{fig:timefreq}.
		
		\begin{figure}
			\centering
			\includegraphics[width=0.75\linewidth]{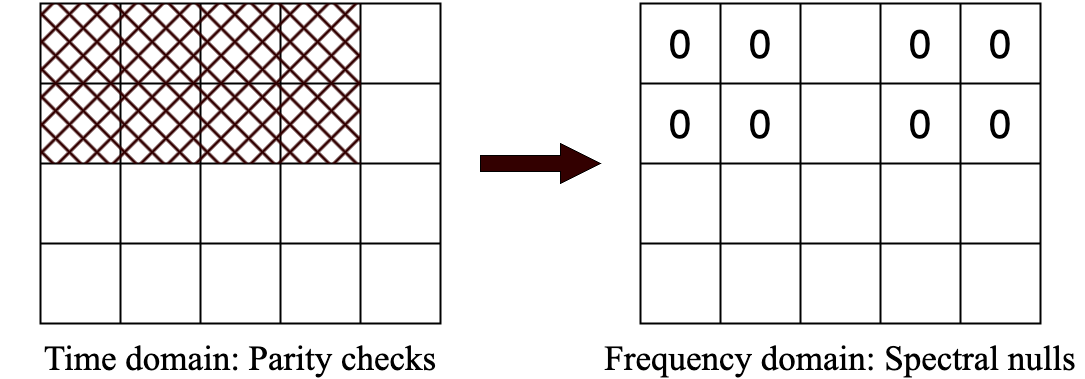}
			\caption{The parity check locations  are organized in contiguous positions in the time domain are hashed and the blank boxes shows the message bit locations. The CZ set identifies the coordinates $(0, 0), (0, 1), (0, 3), (0, 4), (1, 0), (1, 1), (1, 3)$ and $(1, 4)$ as spectral nulls of the given code in Example \ref{NullExamp}.}
			\label{fig:timefreq}
		\end{figure}
	\end{example}

	\section{Decoding Algorithms for 2-D Constacyclic Codes}\label{section5}
	
	Systematic encoding of constacyclic arrays can be done in the time domain based on the generator tensor derived from the parity check tensor of the 2-D constacyclic code. For more details on the encoder design, the reader is referred to Section 5 of our recent paper \cite{CCDS}.\par 
	In this section, our focus will be to derive decoding algorithms for 2-D constacyclic codes. The decoding algorithm for 2-D constacyclic codes involves: (a) identifying the error locations, and (b) the error values at those locations. 
	
	\subsection{Error-detections and their locations} 
	
	Let $r$ be the received array of size $M\times N$ over $\mathbb{F}_q$ and its bi-variate polynomial form is 
	\begin{equation}
		r(x, y)=\sum_{i=0}^{M-1}\sum_{j-0}^{N-1}r_{i,j}x^i y^j.
	\end{equation} 
	Now, we recall a result regarding a bi-variate polynomial to be a codeword from \cite{CCDS}.
	\begin{theorem}\cite[Theorem 3.1]{CCDS}\label{thm2Dconsta}
		Let $\mathcal{C}$ be a 2-D $(\lambda_1, \lambda_2)$-constacyclic code of area $M \times N$ over $\mathbb{F}_q$. Then, a polynomial $h(x, y)$ in $\mathcal{P}[\Omega]$ is a codeword of $\mathcal{C}$ if and only if $h(a, b)= 0$ for each $(a, b) \in {V}_{c}$.
	\end{theorem}
	Let $r$ be the received array. By using the above result, if $r(a, b)=0$ for all $(a, b)\in {V}_{c}$, then $r(x, y)$ is a codeword. If $r(a, b)\neq 0$ for some $(a, b)\in {V}_{c}$ then $r(x,y)$ is not a codeword, that is, an error is detected. In this case, it is equivalent to say,
	\begin{equation*}
		r(a,b) = (c(x,y)+e(x,y))\vert_{(a,b)}=e(a,b)\neq 0,
	\end{equation*}
	for $(a,b) \in {V}_{c}$, that is, $e(x,y)$ should not vanish on ${V}_{c}$, implying that the error pattern $e(x,y)$ is not a codeword.
	\begin{theorem}\label{2Dto1D}
		Let $\mathcal{C}$ be a 2-D $(\lambda_1, \lambda_2)$-constacyclic code of area $M\times N$ over $\mathbb{F}_q$. Suppose $W_i \subseteq \mathbb{F}_q^{N}$ denotes the set of $i^{\rm th}$-row of all the codewords of $\mathcal{C}$. Then $W_i$ is a 1-D row $\lambda_2$-constacyclic code of length $N$ over $\mathbb{F}_q$ for each $i=0, 1, \dots, M-1$.
		\begin{proof}
			Firstly, we will prove each $W_i$ is a subspace of the vector space $\mathbb{F}_q^{N}$. Let $\alpha, \beta \in \mathbb{F}_q$ and $u_{r_i}, u_{r_i}' \in W_i$. Then by construction of $W_i$, $\exists$ $u, u' \in \mathcal{C}$. Since $\mathcal{C}$ is a linear code, $\alpha u + \beta u' \in \mathcal{C}$. Then by definition, $i^{\rm th}$-row of the array $\alpha u + \beta u'$ is $\alpha u_{r_i} + \beta u'_{r_i} \in W_i$. This proves $W_i$ is a linear code for $i=0, 1, \dots, M-1$.\\
			Now, we will prove each $W_i$ is invariant under right $\lambda_2$-shift. Let $u_{r_i}\in W_i$ be an arbitrary element. Then $\exists$ $u\in \mathcal{C}$, by construction $W_i$. Since $\mathcal{C}$ is a $(\lambda_1, \lambda_2)$-constacyclic code, then its right row $\lambda_2$-shift $u_{\lambda_2} \in \mathcal{C}$. Then the $i^{\rm th}$-row of the codeword $u_{\lambda_2}$ is ${u_{\lambda_2}}_{r_i}\in W_i$, which is the right $\lambda_2$-shift of $u_{r_i}$. This proves that $W_i$ is a $\lambda_2$-constacyclic code of length $N$ over $\mathbb{F}_q$ for $i=0, 1, \dots, M-1$.
		\end{proof}
	\end{theorem}
	\begin{theorem}\label{2Dto1Dcol}
				Let $\mathcal{C}$ be a 2-D $(\lambda_1, \lambda_2)$-constacyclic code of area $M\times N$ over $\mathbb{F}_q$. Suppose $Z_j \subseteq \mathbb{F}_q^{M}$ denotes the set of $j^{\rm th}$-column of all the codewords of $\mathcal{C}$. Then $Z_j$ is a 1-D row $\lambda_1$-constacyclic code of length $M$ over $\mathbb{F}_q$ for each $j = 0, 1, \dots, N-1$.
				\begin{proof}
				The argument used to show Theorem \ref{2Dto1D} can also be used to prove this.
				\end{proof}
	\end{theorem}
	
	Now, the next task is to find the locations of these errors. Our central idea for finding the error locations would be to iterate over all rows and columns and check the syndrome by fixing one of the variables $x=1$ or $y=1$ in $R(x,y)$, respectively. The intersection of the error locations of the rows and columns provides the 2-D error coordinates. This idea is summarized as follows.
	
	\begin{theorem}\label{thm2DconstaSpect}
		Consider a 2-D $(\lambda_1, \lambda_2)$-constacyclic code of area $M \times N$ over $\mathbb{F}_q$. Let $r$ be the received array and $r(x, y)$ be its polynomial form in time-domain. Assume that $r(a, b)\neq 0$ for some $(a, b)\in {V}_{c}$. Let $(a, b)\in {V}_{c}$ be an arbitrary element.
		\begin{enumerate}
			\item \textbf{Error locations in a row:}  For each $i=0, 1, \dots, M-1$, consider the following polynomial
			\begin{equation*}
				r_{i}(y):=\sum_{j=0}^{N-1}r_{i,j}y^j,
			\end{equation*}
			where $r_{i,j}$ is the $(i,j)^{\rm th}$ element of the received array $r$.\\
			If $r_{i}(b)=0$ then there is no error in the $i^{\rm th}$-row, and if $r_{i}(b)\neq 0$ then there is an error in the $i^{\rm th}$-row of the received array $r$ in time-domain.
			
			\item \textbf{Error locations in a column:} For each $j=0, 1, \dots, N-1$, consider the following polynomial
			\begin{equation*}
				r_{j}(x):=\sum_{i=0}^{M-1}r_{i,j}x^i,
			\end{equation*}
			where $r_{i,j}$ is the $(i,j)^{\rm th}$ element of the received array $r$.\\
			If $r_{j}(a)=0$ then there is no error in the $j^{\rm th}$-column, and if $r_{j}(a)\neq 0$ then there is an error in the $j^{\rm th}$-column of the received array $r$.
		\end{enumerate}
		\begin{proof}
			Let $\mathcal{C}$ be a 2-D $(\lambda_1, \lambda_2)$-constacyclic code of area $M \times N$ over $\mathbb{F}_q$. Consider a codeword $c(x, y)\in \mathcal{C}$. The evaluation of $c(a,b)=0$ for $(a,b)\in {V_c}$ over the bi-index implies that the evaluation of $c_i(b)=0$ holds true over every fixed row location $i$ since each row is a valid codeword for a 1-D constacyclic code (by using Theorem \ref{2Dto1D} and 1-D version of Theorem \ref{thm2Dconsta}). Let $r$ be the received array. Then $r(x, y) = c(x,y) + e(x,y)$. For a fixed row location $i$, let $r_{i}(y) = \sum_{j=0}^{N-1}r_{i,j}y^j$. Since 
			\begin{equation}\label{eqS}
				r_{i}(y)=c_{i}(y)+e_{i}(y),
			\end{equation}
			by evaluating  Eq. \ref{eqS} at $(a,b) \in V_c$, for a fixed row location $i$, we have 
			\begin{equation*}
				r_{i}(b)=e_{i}(b).
			\end{equation*}
			Thus, the syndromes obtained by evaluating $r_i(b)$ can yield error locations along the rows.\\ 
			Similarly, we can fix the position of a column $j$ and evaluate the polynomials received along the columns, reflecting the same idea we have for the rows, by using Theorem \ref{2Dto1Dcol}. This completes the proof. 
		\end{proof}
	\end{theorem}
	\begin{remark}
		Based on the above discussion, we get the required set $\mathbb{E}$ of error locations $(i, j)$ in the received array~$r$, where $r_i(b)\neq 0$ and $r_j(a)\neq 0$ for $(a, b)\in V_c$. In this way, by iterating over every row and column index of the 2-D received array, we can get the required error locations of any random error or any other error patterns present in the received array within the error detection capability. 
	\end{remark}
	
	\begin{algorithm}
		\caption{Error Detection and Location in a 2-D Constacyclic Code over $\mathbb{F}_{q}$}
		\label{algo:1}
		\begin{algorithmic}[1]
			\STATE \textbf{Input:} Received array $r$
			
			\STATE \textbf{Step 1: Error detection}
			\STATE Evaluate $r(x,y)$ at each element of the CZ set $V_c$
			\IF {$r(a,b) = 0$ for all $(a, b) \in V_c$}
			\STATE No errors detected. Terminate the program.
			\ELSE
			\STATE An error detected. Proceed to \textbf{Step 2}.
			\ENDIF
			\STATE \textbf{Step 2: Error locations in a row}
			\FOR{each row $i=0$ to $M-1$}
			\STATE Evaluate $r_i(b) = \sum_{j=0}^{N-1} r_{i,j} \cdot b^{j}$ for all $(a,b) \in V_c$
			\ENDFOR
			\STATE \textbf{Step 3: Identify rows with errors}
			\STATE $\mathbb{E}_{r} \gets \{ i \mid r_i(b) \neq 0 \text{ where } (a,b)\in V_c \}$
			\STATE \textbf{Step 4: Error locations in a column}
			\FOR{each column $j=0$ to $N-1$}
			\STATE Evaluate $r_j(a) = \sum_{i=0}^{M-1} r_{i,j} \cdot a^{i}$ for all $(a,b)\in V_c$
			\ENDFOR
			\STATE \textbf{Step 5: Identify columns with errors}
			\STATE $\mathbb{E}_{c} \gets \{ j \mid r_j(a) \neq 0 \text{ where } (a,b)\in V_c\}$
			\STATE \textbf{Output:} $\mathbb{E}=\{(i, j) \vert \, i \in \mathbb{E}_{r}$ and $j \in \mathbb{E}_{c}\}$.
		\end{algorithmic}
	\end{algorithm}

	For completeness, we emphasize in this section the error-detection and correction capabilities of a $K$-dimensional 2-D linear code of area $M\times N$ having the number of parity bits is equal to $|V_c| = MN-K$ (see \cite{CCDS}).\par  
	From the Singleton bound, the minimum Hamming distance of a 2-D linear code is bounded above by
	\begin{equation}\label{Singlton}
		d_{\text{min}}\leq \vert V_c \vert +1.
	\end{equation} 
	If the minimum Hamming distance of a 2-D linear code attains the upper bound in Eq. \eqref{Singlton} then the code is  said to be 2-D maximum distance separable (MDS) code. The following result is well-known in the literature \cite{huffman}.
	\begin{theorem}\cite{huffman}\label{DetectCorrect}
		Let $C$ be a linear code over $\mathbb{F}_q$. If $C$ has minimum distance $d_{\text{min}}$,  then $C$ can detect upto $d_{\text{min}}-1 $ errors and correct upto $\lfloor \frac{d_{\text{min}}-1}{2} \rfloor$ errors in any codeword.
	\end{theorem}

	Suppose that there are $t$-errors that occurred during transmission. Then Theorem \ref{thm2DconstaSpect} gives the error locations set for the error pattern $e(x,y)$ as follows:
	\begin{equation}\label{ErrorSetE}
		\mathbb{E}=\{(i_k, j_k)\, \vert \, k=1, 2, \dots, t\}.
	\end{equation}
	Let the error pattern polynomial be
	\begin{equation}\label{error-poly}
		e(x,y)= \sum_{k=1}^{t}e_{i_k, j_k}x^{i_k}y^{j_k}.
	\end{equation}
	By using Theorem \ref{DetectCorrect}, we have
	\begin{equation}\label{Eq41}
		t\leq \Big\lfloor \frac{\vert V_c\vert}{2} \Big\rfloor.
	\end{equation} 
	
	Once the error locations are identified, we need to obtain the error values. We shall first describe a method based on exhaustive search. Thius will be followed by efficient decoding methods to obtain the error values. 
	
	\subsection{Decoding based on exhaustive search over the error locations set: Method I}
	Suppose we have the set of error locations $\mathbb{E}$. Let $r$ be the received array of size $M \times N$ over $\mathbb{F}_q$. Then 
	\begin{equation*}
		r = c + e,
	\end{equation*}
	where $c$ is a codeword and $e$ is an error pattern. Algorithm \ref{algo:1} gives the error location set $\mathbb{E} = \{(i_1,j_1), (i_2,j_2), \dots, (i_t,j_t)\}$ in the error pattern array $e=[e_{i,j}]$. Then
	\begin{equation}
		\begin{cases}
			e_{i, j} \in \mathbb{F}_q^{\ast} \text{ if } (i, j)\in \mathbb{E},\\ 
			e_{i, j} = 0 \text{ if } (i, j)\in \mathbb{E}^c\\ 
		\end{cases}
	\end{equation}
	Let $A$ and $B$ be two arrays of same size. Consider an inner product of $A$ and $B$ denoted by $A \odot B$, defined as the sum of the component-wise product of elements of $A$ and $B$.
	Let $e_{\text{new}}$ be an initial error pattern, where nonzero coefficients on the position $\mathbb{E}$ are arbitrary elements from $\mathbb{F}_q$. As we have $\vert \mathbb{E} \vert =t$, therefore we have $q^t$ choices of error pattern $e_{\text{new}}$. So we perform a search among these choices and calculate syndrome $H\odot e_{new}$ each time and stop the search when the syndrome $H \odot e_{\text{new}}= H \odot r$, yielding the required error pattern $e_{\text{new}}$. Then $c=r-e_{\text{new}}$, gives the required codeword. Since we have $q^t$ choices of error pattern $e_{\text{new}}$, hence the required error array $e_{\text{new}}$ must exist and decoding must succeed within $q^{t}$ trials.
	\par
	Algorithm \ref{Method III} present the above randomized search decoding. The reader must note that the proposed algorithm can handle random errors as well as treating cluster errors as random errors up to a correctable distance.
	\begin{algorithm}[H]
		\caption{Decoding Based on Exhaustive Search over the Error Locations Set}\label{Method III}
		\begin{algorithmic}[1]
			\REQUIRE Check tensor array $H$, received array $r$ of size $M \times N$, error location set $\mathbb{E}$
			\ENSURE Corrected codeword $c$
			
			\STATE Compute the target syndrome: $S_{\text{req}} \gets H \cdot r$
			
			\STATE Let $\mathbb{E} = \{(i_1,j_1), (i_2,j_2), \dots, (i_t,j_t)\}$ where $t = |\mathbb{E}|$
			\STATE Let $\mathcal{F} \gets \mathrm{GF}(q)^t$ be the set of all $q^t$ possible error value combinations over positions in $\mathbb{E}$
			
			\FORALL{$(a_1, a_2, \dots, a_t) \in \mathcal{F}$}
			\STATE Initialize $e_{\text{new}}$ as a zero array of size $M \times N$
			\FOR{$k = 1$ to $t$}
			\STATE Set $e_{\text{new}}(i_k, j_k) \gets a_k$
			\ENDFOR
			
			\STATE Compute syndrome: $S_{\text{new}} \gets H \cdot e_{\text{new}}$
			
			\IF{$S_{\text{new}} = S_{\text{req}}$}
			\STATE Compute corrected codeword: $c \gets r - e_{\text{new}}$
			\RETURN $c$
			\ENDIF
			\ENDFOR
		\end{algorithmic}
	\end{algorithm}
	
	Now we introduce two more efficient decoding methods in the subsequent subsections.

	We recall a result from \cite{hoffman1973algebra}.
	
	\begin{theorem}\cite{hoffman1973algebra}\label{RankThm}
		Let $AX = B$ be a system of non-homogeneous linear equations over a finite field $\mathbb{F}_q$. Then:
		\begin{enumerate}
			\item If ${\rm rank}(A) = {\rm rank}(A\vert B) = \#$ unknowns, then the system has a unique solution in $\mathbb{F}_q$;
			\item If ${\rm rank}(A) = {\rm rank}(A\vert B) < \#$ unknowns, the system has $q^{t - {\rm rank}(A)}$ solutions, where $t$ is the number of unknowns;
			\item If ${\rm rank}(A) \neq {\rm rank}(A\vert B)$, the system is inconsistent (has no solution) in $\mathbb{F}_q$.
		\end{enumerate}
	\end{theorem}
	
	Let $\mathcal{C}$ be a 2-D constacyclic code of area $M\times N$ over $\mathbb{F}_q$ in time-domain. Let $r$ be a received array of area $M\times N$. We can write in bi-variate polynomial form
	\begin{equation}\label{eq36}
		r(x,y)=c(x,y)+e(x,y).
	\end{equation}

	\subsection{Decoding in the time-domain: Method II}\label{Method:I}
	This subsection presents a decoding method to determine the error pattern $e(x,y)$, which subsequently leads to the recovery of the desired codeword $c(x,y)$.\par 
	Assume that there were $t$-errors due to the channel noise. From Theorem \ref{thm2DconstaSpect} we obtain the error locations set for the error pattern $e(x,y)$ as follows:
	\begin{equation}\label{ErrorSetE}
		\mathbb{E}=\{(i_k, j_k)\, \vert \, k=1, 2, \dots, t\}.
	\end{equation}
	Let the error pattern polynomial be
	\begin{equation}\label{error-poly}
		e(x,y)= \sum_{k=1}^{t}e_{i_k, j_k}x^{i_k}y^{j_k}.
	\end{equation}
	Note that
	\begin{equation}\label{Eq41}
		t\leq \Big\lfloor \frac{\vert V_c\vert}{2} \Big\rfloor 
	\end{equation} 
	by using Theorem \ref{DetectCorrect}. Let $\vert V_c\vert =s$. By Theorem \ref{thm2Dconsta} and Eq. \eqref{eq36}, we have the following non-homogeneous system of linear equations
	\begin{equation}\label{SystemEq}
		\begin{split}
			r(a_i, b_i) &= c(a_i, b_i) + e(a_i, b_i)\\
			&= e(a_i, b_i),
		\end{split}
	\end{equation}
	for all $(a_i, b_i) \in V_c$, where $i= 1, 2, \dots, s$. Equivalently, we have the following non-homogeneous system of linear equations
	\begin{equation}\label{systm}
		AX = B,
	\end{equation}
	where 
	\begin{equation*}
		A= \begin{pmatrix}
			a_1^{i_1}b_1^{j_1} & a_1^{i_2}b_1^{j_2} & \cdots & a_1^{i_t}b_1^{j_t}\\
			a_2^{i_1}b_2^{j_1} & a_2^{i_2}b_2^{j_2} & \cdots & a_2^{i_t}b_2^{j_t}\\
			\vdots & \vdots  & \ddots & \vdots\\
			a_{s}^{i_1}b_{s}^{j_1} & a_{s}^{i_2}b_{s}^{j_2} & \cdots & a_{s}^{i_t}b_{s}^{j_t}
		\end{pmatrix}, \quad 
		X= 
		\begin{pmatrix}
			e_{i_1, j_1}\\
			e_{i_2, j_2}\\
			\vdots\\
			e_{i_t, j_t}\\
		\end{pmatrix} \quad \text{ and } \quad
		B=
		\begin{pmatrix}
			r(a_1, b_1)\\
			r(a_2, b_2)\\
			\vdots \\
			r(a_s, b_s)
		\end{pmatrix}.
	\end{equation*}
	Now we have the following result.
	\begin{theorem}\label{NonHomoTD}
		The non-homogeneous system of linear equations \eqref{systm} has a unique solution.
	\end{theorem}
	\begin{proof}
		Let $r(x,y)$ be a received array polynomial.  If possible, let $B=0$ then $r(a_i, b_i)=0$ for all $(a_i, b_i)\in V_c, i=1, 2, \dots, s$. Then by Theorem \ref{thm2Dconsta}, $r(x,y)$ is a codeword, which is a contradiction. Hence the system is non-homogeneous.\\
		We know that  ${\rm rank}(A) \leq {\rm rank}(A\vert B)$. Now, we will prove ${\rm rank}(A) = {\rm rank}(A\vert B)$.\\
		Let ${\rm rank}(A)=\rho$. We know that ${\rm rank}(A)\leq {\rm min}(s, t)=t$, that is, $\rho\leq t$. Let $A_i$ and $\tilde{A_i}$ denote the $i^{\rm th}$ row of $A$ and of the augmented matrix $(A\vert B)$, respectively.\\
		By rearranging the equations of the system \eqref{SystemEq}, we can always make the first $\rho$ rows of $A$ linearly independent. So for the sake of simplicity, we assume that the first $\rho$ rows of the matrix $A$, namely, $\{A_1, A_2, \dots, A_{\rho}\}$, are linearly independent, and the last $s-\rho$ rows of $A$ are in ${\rm Span}(A_1, A_2, \dots, A_{\rho})$, that is, $k^{\rm th}$ row of $A$ can be expressed as
		\begin{equation}\label{Eq40}
			A_k=\sum_{i=1}^{\rho}z_iA_i,
		\end{equation}
		where $\rho +1 \leq k \leq s$ and $z_i\in \mathbb{F}_q$. Then by using Eqs. \eqref{SystemEq}, \eqref{systm} and \eqref{Eq40}, for the $k^{\rm th}$-row, we have
		\begin{equation}
			\begin{split}
				e(a_k, b_k) =& \sum_{i=1}^{\rho}z_ie(a_i, b_i),\\
				\implies	r(a_k, b_k)=&\sum_{i=1}^{\rho}z_ir(a_i, b_i),
			\end{split}
		\end{equation}
		where $\rho +1 \leq k \leq s$. Then 
		\begin{equation}\label{Eq42}
			\tilde{A_k}=\sum_{i=1}^{\rho}z_i \tilde{A_i},
		\end{equation}
		where $\rho +1 \leq k \leq s$. Then from Eq. \eqref{Eq42}, we have ${\rm rank}(A\vert B) \leq \rho$, that is, $\rho \leq {\rm rank}(A\vert B) \leq \rho$. Hence, ${\rm rank}(A) = {\rm rank}(A\vert B)$.\\
		Now we will prove that ${\rm rank}(A) = {\rm rank}(A\vert B)=t$.\\ 
		We have ${\rm rank}(A)\leq t$. If possible, assume that ${\rm rank}(A)< t$. By Theorem \ref{RankThm}(2), the system \eqref{systm} has $q^{t-{\rm rank}(A)}$ solutions in $\mathbb{F}_q$.\\ 
		Let $X^{(1)}$ and $X^{(2)}$ be two distinct solutions of the system \eqref{systm}. Then this will give two distinct error patterns $e^{(1)}(x, y)$ and $e^{(2)}(x, y)$. From Eq. \eqref{eq36}, we have $e^{(1)}(x, y)=r(x,y)-c(x,y)=e^{(2)}(x, y)$, which is a contradiction.\\
		Hence, ${\rm rank}(A) = {\rm rank}(A\vert B)=t$, the number of unknowns. This implies that the system \eqref{systm} has a unique solution. 
	\end{proof}
	Thus, the system of equations \eqref{systm} has a unique solution, leading to a unique error pattern $e(x,y)$. Thus
	\begin{equation*}
		c(x,y)= r(x,y)-e(x,y)
	\end{equation*}
	is the required decoded codeword.
	\subsection{Decoding using finite field Fourier transforms: Method III}\label{Method:II}
	This subsection presents another decoding method using FFFT. In this process, we determine the codeword spectrum $C(x,y)$, which subsequently leads to the recovery of the desired codeword $c(x,y)$ by taking the inverse Finite Field Fourier Transform (IFFFT). We can cleverly exploit \textit{sparsity} due to \textit{duality} while dealing with frequency-domain properties of the code. \par
	\textbf{Advantage of \textit{Method III}:} When $\vert \Omega \setminus \Pi \vert <<\vert \Pi\vert$, then \textit{Method:III} will work more efficiently over \textit{Method: II}, as we need to deal with a non-homogeneous system of linear equations with a lesser number of unknowns.\par 
	Let $\mathcal{C}$ be a 2-D constacyclic code of area $M\times N$ over $\mathbb{F}_q$. Let $r$ be a received array of area $M\times N$ over $\mathbb{F}_q$. Here $r=c+e$, where $c\in \mathcal{C}$, and $e$ is an error pattern outside $\mathcal{C}$, that is, it is a detectable error.  Then its image under FFFT is
	\begin{equation}
		\begin{split}
			{\rm FFFT}(r) & = {\rm FFFT}(c) + {\rm FFFT}(e) \\
			R & = C + E
		\end{split}
	\end{equation}
	Note that all the arrays $R$, $C$ and $E$ are of the same area $M\times N$, but over an extended field.\\
	We can write it in bi-variate polynomial form
	\begin{equation}\label{eq47}
		R(x,y)=C(x,y)+E(x,y).
	\end{equation}
	Using Lemma \ref{spectralNullLemma}, we have $C_{\theta, \phi}=0$ for each pair $(\theta, \phi)$ when $(\gamma \zeta_1^{\theta}, \beta \zeta_2^{\phi})\in V_c$. We collect all such pairs $(\theta, \phi)$ and construct the set $\widetilde{\mathbb{E}}$. The set $\widetilde{\mathbb{E}}$, where the spectrum $C$ is nulling, is completely determined by Lemma \ref{spectralNullLemma}. So we have
	\begin{equation}
		C_{\theta, \phi}=0 \quad \forall \quad (\theta, \phi)\in \widetilde{\mathbb{E}}.
	\end{equation}
	Observe that $\vert \widetilde{\mathbb{E}} \vert = \vert \Pi \vert  =\vert V_c\vert \implies \vert  \widetilde{\mathbb{E}}^c\vert = MN- \vert V_c\vert = s'$ (say).\\
	Let 
	\begin{equation}
		\widetilde{\mathbb{E}}^c=\{(\theta_k, \phi_k) \, \vert \, k=0, 1, \dots, s' \}.
	\end{equation} 
	Now the problem of finding the spectrum polynomial $C$ boils down to find the entries of $C$ only at $(\theta_k, \phi_k)\in \widetilde{\mathbb{E}}^c$.\\
	Let the spectrum polynomial over the extended field be
	\begin{equation}\label{eq51}
		C(x,y)= \sum_{k=1}^{s'}C_{\theta_k, \phi_k}x^{\theta_k}y^{\phi_k}.
	\end{equation}
	Let $\vert \mathbb{E}^c \vert = t'$. Note that 
	\begin{equation}
		e_{i, j}=0 \quad \forall \quad (i, j)\in \mathbb{E}^c=\Omega\setminus \mathbb{E},
	\end{equation} 
	where the set $\mathbb{E}$ of error locations of the error pattern $e(x,y)$ is given by Eq. \eqref{ErrorSetE}.  Let
	\begin{equation}
		\mathbb{E}^c = \{(i_l, j_l) \, \vert \, l=1, 2, \dots, t' \}.
	\end{equation}
	Consequently, for spectrum error polynomial $E(x,y)$, we have 
	\begin{equation}\label{Eq54}
		E(\zeta_1^{-i_l}, \zeta_2^{-j_l})=0 \quad \forall \quad l=1, 2, \dots, t',
	\end{equation}
	by Theorem \ref{th6}(2).
	Here, we have $\vert \mathbb{E}\vert = t$, $\vert \mathbb{E}^c \vert = t'$, $\vert V_c\vert =s$, $\vert \widetilde{\mathbb{E}}^c \vert =MN-\vert V_c\vert =s'$ and from Eq. \eqref{Eq41}
	\begin{equation}
		\begin{split}
			& t\leq \Big\lfloor \frac{\vert V_c\vert}{2} \Big\rfloor  \leq \vert V_c\vert \\
			\implies & MN-\vert V_c\vert  \leq MN-t \\
			\implies & \vert \widetilde{\mathbb{E}}^c\vert  \leq \vert \mathbb{E}^c\vert \\
			\implies & s'\leq t'.
		\end{split}
	\end{equation}
	By Eqs. \eqref{eq47}, \eqref{eq51} and \eqref{Eq54}, we have the following non-homogeneous system of linear equations
	\begin{equation}\label{SystemEqFD}
		\begin{split}
			R(\zeta_1^{-i_l}, \zeta_2^{-j_l}) &= C(\zeta_1^{-i_l}, \zeta_2^{-j_l}) + E(\zeta_1^{-i_l}, \zeta_2^{-j_l})\\
			&= C(\zeta_1^{-i_l}, \zeta_2^{-j_l}),
		\end{split}
	\end{equation}
	for all $(i_l, j_l) \in \mathbb{E}^c$, where $l= 1, 2, \dots, t'$. Equivalently, we have the following non-homogeneous system of linear equations (by an abuse of notation we are using the same symbol as we use in Eq. \eqref{systm})
	\begin{equation}\label{systmFD}
		AX = B,
	\end{equation}
	where 
			$$A= \begin{pmatrix}
				{\zeta_1}^{-i_1 \theta_1} {\zeta_2}^{-j_1 \phi_1}  & {\zeta_1}^{-i_1 \theta_2} {\zeta_2}^{-j_1 \phi_2} & \cdots & {\zeta_1}^{-i_1 \theta_{s'}} {\zeta_2}^{-j_1 \phi_{s'}}\\
				{\zeta_1}^{-i_2 \theta_1} {\zeta_2}^{-j_2 \phi_1}  & {\zeta_1}^{-i_2 \theta_2} {\zeta_2}^{-j_2 \phi_2} & \cdots & {\zeta_1}^{-i_2 \theta_{s'}} {\zeta_2}^{-j_2 \phi_{s'}}\\
				\vdots & \vdots  & \ddots & \vdots\\
				{\zeta_1}^{-i_{t'} \theta_1} {\zeta_2}^{-j_{t'} \phi_1}  & {\zeta_1}^{-i_{t'} \theta_2} {\zeta_2}^{-j_{t'} \phi_2} & \cdots & {\zeta_1}^{-i_{t'} \theta_{s'}} {\zeta_2}^{-j_{t'} \phi_{s'}}\\
			\end{pmatrix},$$ 
			$$X= 
			\begin{pmatrix}
				C_{\theta_1, \phi_1}\\
				C_{\theta_2, \phi_2}\\
				\vdots\\
				C_{\theta_{s'}, \phi_{s'}}\\
			\end{pmatrix} \quad \text{ and } \quad
			B=
			\begin{pmatrix}
				R(\zeta_1^{-i_1}, \zeta_2^{-j_1})\\
				R(\zeta_1^{-i_2}, \zeta_2^{-j_2})\\
				\vdots \\
				R(\zeta_1^{-i_{t'}}, \zeta_2^{-j_{t'}})
			\end{pmatrix}.$$
	Now, we have the following result.
	\begin{theorem}
		The non-homogeneous system of linear equations \eqref{systmFD} has a unique solution.
	\end{theorem}
	\begin{proof}
		The result directly follows from the arguments presented in the proof of Theorem \ref{NonHomoTD}.
	\end{proof}
	Thus, the system of equations \eqref{systmFD} has a unique solution, leading to a unique codeword spectrum $C(x,y)$. Thus
	\begin{equation*}
		c= {\rm IFFFT}(C)
	\end{equation*}
	is the required decoded codeword in time-domain.
	\begin{remark}
		One can opt any of the following methods to solve the system of linear equations in Method \ref{Method:I} and Method \ref{Method:II}:
		\begin{enumerate}
			\item Gaussian elimination over $\mathbb{F}_q$;
			\item Matrix inversion;
			\item LU decomposition over $\mathbb{F}_q$;
			\item Using the null space (homogeneous + particular solution), etc.
		\end{enumerate}
	\end{remark}

	\subsection{Numerical examples for the above Methods} In this subsection, we will provide a few examples that demonstrate the aforesaid decoding algorithms.
	\begin{example}
		Let $\mathbb{F}_{81}=\mathbb{F}_3(\alpha)$, where $\alpha$ is a root of a primitive polynomial $x^4+x+2$. Suppose $\zeta_1=\alpha^{20}, \zeta_2=\alpha^{16}$ are $4^{\rm th}$ and $5^{\rm th}$ roots of unity, respectively, and $\gamma =\alpha^{10}$ and $\beta =\alpha^{8}$ are $4^{\rm th}$ and $5^{\rm th}$ roots of $\lambda_1=2$ and $\lambda_2=2$, respectively. Let $\mathcal{C}$ be a 2-D $(2, 2)$-constacyclic code of area $4\times 5$ over $\mathbb{F}_{3}$ having CZ set 
		$$V_{c}=\{(\gamma, \beta), (\gamma^3, \beta^3), (\gamma, 2\beta^4), (\gamma^3, 2\beta^2); (\gamma, \beta^3), (\gamma^3, 2\beta^4), (\gamma, 2\beta^2), (\gamma^3, \beta)\}.$$
		Let the received array be 
		\begin{equation*}
			r =
			\begin{pmatrix}
				0 & 2  & 0 & 0 & 0\\
				0 & 0  & 0 & 0 & 0\\
				0 & 0  & 0 & 0 & 0\\
				0 & 0  & 0 & 0 & 0
			\end{pmatrix}.
		\end{equation*}
		
		From Theorem \ref{thm2DconstaSpect}, we have $r_{0}(\beta)=2\beta \neq 0$, $r_{1}(\beta)=0$, ${r_2}(\beta)=0$ and ${r_3}(\beta)=0 \implies $ there is an error in the $0^{\rm th}$-row.\\
		Also, we have ${r_0}(\gamma)= 0$, ${r_1}(\gamma)=2 \neq 0$, ${r_2}(\gamma)=0$ and ${r_3}(\gamma)=0 $, ${r_4}(\gamma)=0 \implies $ there is an error in the $1^{\rm st}$-column. Hence, the set of error locations is $\mathbb{E}=\{(0, 1)\}$. So we can assume that the required error pattern is $e(x,y)=e_{0, 1}y$. We have $r(x,y)=2y$. From Eq. \eqref{systm}, we have 
		\begin{equation*}
			\begin{pmatrix}
				\beta \\
				\beta^3\\
				2\beta^4\\
				2\beta^2\\
				\beta^3\\
				2\beta^4\\
				2\beta^2\\
				\beta
			\end{pmatrix}
			\begin{pmatrix}
				e_{0,1}
			\end{pmatrix}=
			\begin{pmatrix}
				2\beta \\
				2\beta^3\\
				\beta^4\\
				\beta^2\\
				2\beta^3\\
				\beta^4\\
				\beta^2\\
				2\beta
			\end{pmatrix}.
		\end{equation*}
		Note that ${\rm rank}(A)={\rm rank}(A\vert B)=1$, number of unknowns. 
		Therefore, we have a unique solution. On solving the above system, we get $e_{0,1}=2$, that is, $e(x,y)=2y$. Thus, $c(x,y)=r(x,y)-e(x,y)=0$. Hence the transmitted codeword was 
		\begin{equation*}
			c =
			\begin{pmatrix}
				0 & 0  & 0 & 0 & 0\\
				0 & 0  & 0 & 0 & 0\\
				0 & 0  & 0 & 0 & 0\\
				0 & 0  & 0 & 0 & 0
			\end{pmatrix}.
		\end{equation*}
	\end{example}
	
	\begin{example}
		Let $\mathbb{F}_{81}=\mathbb{F}_3(\alpha)$, where $\alpha$ is a root of a primitive polynomial $x^4+x+2$. Suppose $\zeta_1=\alpha^{20}, \zeta_2=\alpha^{16}$ are $4^{\rm th}$ and $5^{\rm th}$ roots of unity, respectively, and $\gamma =\alpha^{10}$ and $\beta =\alpha^{8}$ are $4^{\rm th}$ and $5^{\rm th}$ roots of $\lambda_1=2$ and $\lambda_2=2$, respectively. Let $\mathcal{C}$ be a 2-D $(2, 2)$-constacyclic code of area $4\times 5$ over $\mathbb{F}_{3}$ having CZ set 
		$$V_{c}=\{(\gamma, \beta), (\gamma^3, \beta^3), (\gamma, 2\beta^4), (\gamma^3, 2\beta^2); (\gamma, \beta^3), (\gamma^3, 2\beta^4), (\gamma, 2\beta^2), (\gamma^3, \beta)\}.$$
		Let the received array be 
		\begin{equation*}
			r =
			\begin{pmatrix}
				0 & 2  & 0 & 0 & 0\\
				0 & 1  & 0 & 0 & 0\\
				0 & 0  & 0 & 0 & 0\\
				0 & 0  & 0 & 0 & 0
			\end{pmatrix}.
		\end{equation*}		
		From Theorem \ref{thm2DconstaSpect}, we have ${r_0}(\beta)=2\beta \neq 0$, ${r_1}(\beta)=\beta \neq 0$, ${r_2}(\beta)=0$ and ${r_3}(\beta)=0 \implies $ there is an error in the $0^{\rm th}$-row and $1^{\rm st}$-row.\\
		Also, we have ${r_0}(\gamma)= 0$, ${r_1}(\gamma)=2 + \gamma \neq 0$, ${r_2}(\gamma)=0$ and ${r_3}(\gamma)=0 $, ${r_4}(\gamma)=0 \implies $ there is an error in the $1^{\rm st}$-column. Hence, the set of error locations is $\mathbb{E}=\{(0, 1), (1, 1)\}$. So we assume that the required error pattern is $e(x,y)=e_{0, 1}y + e_{1, 1}xy$. We have $r(x,y)=2y + xy$. From Eq. \eqref{systm}, we have 
		\begin{equation*}
			\begin{pmatrix}
				\beta & \gamma \beta \\
				\beta^3 & \gamma^3 \beta^3\\
				2\beta^4 & 2\gamma \beta^4\\
				2\beta^2 & 2\gamma^3 \beta^2\\
				\beta^3 & \gamma \beta^3\\
				2\beta^4 & 2\gamma^3 \beta^4\\
				2\beta^2 & 2\gamma \beta^2\\
				\beta & \gamma^3 \beta
			\end{pmatrix}
			\begin{pmatrix}
				e_{0,1}\\
				e_{1, 1}
			\end{pmatrix}=
			\begin{pmatrix}
				2\beta + \gamma \beta \\
				2\beta^3 + \gamma^3 \beta^3\\
				\beta^4 + 2\gamma \beta^4\\
				\beta^2 + 2\gamma^3 \beta^2\\
				2\beta^3 + \gamma \beta^3\\
				\beta^4 + 2\gamma^3 \beta^4\\
				\beta^2 + 2\gamma \beta^2\\
				2\beta + \gamma^3 \beta
			\end{pmatrix}.
		\end{equation*}
		Note that ${\rm rank}(A)={\rm rank}(A\vert B)=2$, number of unknowns. Therefore, we have a unique solution.  On solving the above system, we get $e_{0,1}=2$, $e_{1, 1} =1$ that is, $e(x,y)=2y + xy$. Thus, $c(x,y)=r(x,y)-e(x,y)=0$. Hence the transmitted codeword was 
		\begin{equation*}
			c =
			\begin{pmatrix}
				0 & 0  & 0 & 0 & 0\\
				0 & 0  & 0 & 0 & 0\\
				0 & 0  & 0 & 0 & 0\\
				0 & 0  & 0 & 0 & 0
			\end{pmatrix}.
		\end{equation*}
	\end{example}

	\section{Conclusions}\label{section6}
	
	We derived the spectral domain properties of 2-D constacyclic coded arrays over $\mathbb{F}_q$ by viewing them through the lens of common zero sets, generalizing the results for binary cyclic codes originally proposed by Imai. We also proposed efficient decoding methods for correcting random errors within the design distance of 2-D constacyclic codes using syndrome computations in both the time and frequency domains. Our proposed ideas are practically useful for next-generation 2-D bar codes and in data storage devices that work with native 2-D codes without requiring rastering of 1-D codes and other overheads.

	\section*{Acknowledgement}
	Vidya Sagar and Shikha Patel are supported by the Institute of Eminence (IoE) Postdoctoral Fellowship at IISc for carrying out this research.
	\bibliography{funbibfile}
\end{document}